\documentclass[12pt]{article}
\usepackage[top=1in,nohead,bottom=1in]{geometry}
\usepackage{graphicx}
\usepackage{amsmath, amssymb, amsfonts, amsthm, float,setspace}  
\usepackage{enumerate, color, framed, float, multirow}
\usepackage{comment, longtable, caption, subcaption, appendix,hyperref}
\usepackage{mathtools}
\usepackage[round]{natbib}
\usepackage[titles]{tocloft}
\usepackage{setspace, parskip}

\usepackage{algorithm}
\usepackage[noend]{algpseudocode}

\makeatletter
\def\BState{\State\hskip-\ALG@thistlm}
\makeatother

\allowdisplaybreaks
\newcommand{\ds}{\displaystyle}

\newcommand{\X}{\mathcal{X}}

\newcommand{\bb}[1]{\mathbf{#1}}
\newcommand{\bs}[1]{\boldsymbol{#1}}
\newcommand\numberthis{\addtocounter{equation}{1}\tag{\theequation}}

\newtheorem{theorem}{Theorem}
\newtheorem{defi}{Definition}
\newtheorem{propo}{Proposition}
\newtheorem{lemma}{Lemma}
\newtheorem{corollary}{Corollary}

\theoremstyle{remark}

\newtheorem{remark}{Remark}
\newtheorem{assum}{Assumption}

\author{
Sanket Agrawal\\
Department of Statistics\\
University of Warwick
\and
Dootika Vats\\
Department of Mathematics and Statistics\\
Indian Institute of Technology Kanpur\\
\texttt{dootika@iitk.ac.in}
\and
Krzysztof {\L}atuszy{\'n}ski\\
Department of Statistics\\
University of Warwick
\and
Gareth O. Roberts\\
Department of Statistics\\
University of Warwick
}
\title{Optimal Scaling of MCMC Beyond Metropolis}
\date{\today}

\begin{document}

\maketitle
 
\doublespacing
\begin{abstract}
    The problem of optimally scaling the proposal distribution in a Markov chain Monte Carlo algorithm is critical to the quality of the generated samples. Much work has gone into obtaining such results for various Metropolis-Hastings (MH) algorithms. Recently, acceptance probabilities other than MH are being employed in problems with intractable target distributions. There is little resource available on tuning the Gaussian proposal distributions for this situation.  We obtain optimal scaling results for a general class of acceptance functions, which includes Barker's and Lazy-MH.      In particular, optimal values for the Barker's algorithm are derived and found to be significantly different from that obtained for the MH algorithm. Our theoretical conclusions are supported by numerical simulations indicating that when the optimal proposal variance is unknown, tuning to the optimal acceptance probability remains an effective strategy.
\end{abstract}
 
 % \onehalfspacing

\section{Introduction} % (fold)
\label{sec:intro}

%\begin{comment}  New Introduction (26/1)
    Over the past few decades, Markov chain Monte Carlo (MCMC) methods have  become an abundantly popular computational tool, enabling practitioners to conveniently sample from complicated target distributions \citep[see][]{brooks2011,meyn2012, robert2013}.
    % \citep{and:rob:2009,gonccalves2017a,park:har:2018}. 
    This popularity can be attributed to easy-to-implement accept-reject based MCMC algorithms for target densities available only up to a proportionality constant. Here, draws from a proposal kernel are accepted with a certain \textit{acceptance probability}. The choice of the acceptance probability and the proposal kernel can yield varying performances of the MCMC samplers.

    Unarguably, the most popular acceptance probability is Metropolis-Hastings (MH) of \cite{metropolis1953,hastings1970} due to its acknowledged optimality \citep{peskun1973,billera2001}. Efficient implementation of the MH algorithm requires tuning within the chosen family of proposal kernels. For the MH acceptance function, various optimal scaling results have been obtained under assumptions on the proposal and the target distribution. This includes the works of \cite{roberts1997weak,roberts1998,roberts2001,neal:optimal:2006,bedard2008,sher:rob:2009,zanella2017dirichlet,yang2020}, among others.

    % The MH algorithm can be implemented easily as long as the target distribution is tractable up to a proportionality constant. 
    % as efficient sampling tools when the target distribution is available up to a proportionality constant. The comparitive ease of implementation in these situations is due to the dependence of MH acceptance function on current and proposed values only through the ratio of their target densities. When this ratio is analytically tractable, MH yields best results in the class of all such acceptance functions \citep{peskun1973,billera2001}. But, given the complexity usually involved with Bayesian models, it is frequently observed that the required ratio is either intractable or computationally expensive. 

    Despite the popularity of the MH acceptance function, other acceptance probabilities remain practically and theoretically relevant.  Recently, the Barker's acceptance rule \citep{barker1965} and the Lazy-MH \citep[see][]{latuszynski2013} have found use in Bernoulli factory based MCMC algorithms for intractable posteriors \citep{herbei:berliner:2014,gonccalves2017a, gonccalves2017b, smith:2018,vats2020}. Barker's acceptance function has also proven to be optimal with respect to search efficiency \citep{menezes2014} and it guarantees variance improvements for waste-recycled Monte Carlo estimators \citep{del:jour:2009}. Further, a class of acceptance probabilities from  \cite{bedard2008} has been of independent theoretical interest.  We also introduce a new family of \textit{generalized Barker's} acceptance probabilities and present a Bernoulli factory for use in problems with intractable posteriors. 

    % Various approaches have been proposed to handle such situations \citep{mira2001,and:rob:2009, fnl11,banterle2019accelerating}. However, the motivation for the present paper comes from the growing popularity of Barker's acceptance function \citep{barker1965} among the MCMC community as an alternative to MH in these cases. Although, Barker's acceptance probability depends on the same ratio as MH, it has been shown to admit effcient Bernoulli factories, and hence can be viably implemented \citep{vats2020, gonccalves2017a, gonccalves2017b}. That is, the Barker's function, which has been shown to be no worse than 2 times the MH \citep{latuszynski2013}, can be particularly important in situations leading to intractable acceptance ratio. Moreover, Barker's acceptance function also have a better search efficiency then MH which makes a case for it in various optimization problems \citep{menezes2014}. 

    To the best of our knowledge, there are no theoretical and practical guidelines concerning optimal scaling  outside of MH and its variants (although see \citealp{sherlock2015} for a discussion on delayed acceptance MH and \citealp{sherlock2015efficiency,doucet2015efficient,schmon2021large} for analyses pertaining to pseudo-marginal MCMC). We obtain optimal scaling results for a large class of acceptance functions; Barker's, Lazy-MH, and MH are members of this class. 
    % We also discuss a few other acceptance probabilities within this class that have been discussed in the literature.

    % In light of the above developments, it then becomes extremely useful for practitioners working with Barker's acceptance function to have theoretically motivated guidelines concerning optimal scaling of the proposal variance and optimal acceptance rate in higher dimensions. \citep{tierney1994,gilks1995,gelman1996,roberts2004}. So far, such results have been obtained only for MH acceptance functions \citep{roberts1997weak,roberts1998,roberts2001,bedard2008,yang2020} and some modifications of it \citep{sherlock2015}. 
% 
% 
    % In addition to Barker's and MH, $\mathcal{A}$ also includes modifications of MH like the Lazy-MH \citep{latuszynski2013}. Further, this class is closed under complex combinations and thus extends to situations when different acceptance probabilites are combined at the update step. For $\mathcal{A}$, we obtain weak convergence results and an expression for asymptotically optimal acceptance rate under the assumptions of Gaussian proposals and regular target densities. We find that the optimal proposal variance and optimal acceptance rate varies in different cases. 
    
    We restrict our attention to the framework of \cite{roberts1997weak} with a random walk Gaussian proposal kernel and a $d$-dimensional decomposable target distribution.   Similar to MH, our   general class of acceptance functions require the proposal variance to be scaled by $1/d$. We find that, typically, for lower acceptance functions, the optimal proposal variance is larger than the optimal proposal variance for MH, implying the need for larger jumps. For the Barker's acceptance rule, the asymptotically optimal acceptance rate (AOAR) is approximately $0.158$, in comparison to the  $0.234$ rate for MH \citep{roberts1997weak}. Similar AOARs are presented for other acceptances. 
    % While for Lazy-MH, it is just the constant times $0.234$. We also find optimal values for two sequence of functions in $\mathcal{A}$, that may serve as nice approximations to MH. As these functions approach MH, the optimal acceptance in either case approaches $0.234$ (Figure \ref{fig:optimal}). Table \ref{tab:results} summarizes the results of our analysis and is thus a key takeaway from the paper. 

   In Section \ref{sec:acc_func} we  describe our class of acceptance probabilities with the main results  presented in Section \ref{sec:main}. Asymptotically optimal acceptance rate for Barker's and other functions are obtained in Section \ref{sec:barker}. In Section~\ref{sec:sims} we present numerical results  under settings that both do and do not comply with our assumptions.  A trailing discussion on the scaling factor for different acceptance functions and generalizations of our results is provided in the last section. All proofs are in the appendices.

\section{Class of acceptance functions}
\label{sec:acc_func}
    % A key feature describing any accept-reject MCMC algorithm is its \textit{acceptance function} which determines the probability with which a proposed move is accepted. 
    Let $\bs{\pi}$ be the target distribution, with corresponding Lebesgue density $\pi$ and support $\mathcal{X}$ so that an MCMC algorithm aims to generate a $\bs{\pi}$-ergodic Markov chain, $\{X_n\}$. Let $Q$ be a Markov kernel with an associated Lebesgue density $q(x, \cdot)$ for each $x \in \mathcal{X}$. We assume throughout that $q$ is symmetric. Further, let the acceptance probability function be $\alpha(x, y): \mathcal{X} \times \mathcal{X} \to [0,1]$.
    % be a function defined on $\mathcal{X} \times \mathcal{X}$ and taking values in $[0, 1]$. 
    Starting from an $X_0 \in \mathcal{X}$, at the $n$th step, a typical accept-reject MCMC algorithm proposes $y \sim q(X_{n-1}, \cdot)$.  The proposed value is accepted with probability $\alpha(X_{n-1}, y)$, otherwise it is rejected, implying that $X_n = X_{n-1}$. 
    % If the value is accepted, the chain moves to $y$, else it remains at $X_{n-1}$, unless a proposed move is accepted. 
    The acceptance function $\alpha$ is responsible for guaranteeing  $\bs{\pi}$-reversibility and thus $\bs{\pi}$-invariance of the Markov chain.
        
    Let $a \wedge b$ denote $\min(a, b)$, and, $s(x,y) = \pi(y) / \pi(x)$. We define $\mathcal{A}$, the class of acceptance functions for which our optimal scaling results will hold, as follows:
    % We assume the following for $\mathcal{A}$,
    \begin{defi}
    \label{assum:acceptance}
        Each $\alpha \in \mathcal{A}$ is a map $\alpha(x,y): \mathcal{X} \times \mathcal{X} \to [0,1]$ and  for every $\alpha \in \mathcal{A},$ there exists a \textit{balancing} function, $g_{\alpha}: [0, \infty) \to [0, 1]$, such that,
        \begin{align}
            & \alpha(x, y) = g_{\alpha}(s(x, y)), \ \  x, y \in \mathcal{X}, \label{eq:gfunc}\\
            & g_{\alpha}(z) = zg_{\alpha}\left( \dfrac{1}{z}\right), \ \ 0 \le z < \infty, \label{eq:detailed_bal}\\
            & g_{\alpha}(e^z), z \in \mathbb{R} \text{ is Lipschitz continuous.} \label{eq:g_lipschitz}
        \end{align}
    \end{defi}

%%%%%%%
\begin{comment}
    \begin{enumerate}
        \item For each $\alpha \in \mathcal{A},$ $\alpha(x, y) \in [0, 1]$ for all $x, y \in \mathcal{X}$.    
        \item For each $\alpha \in \mathcal{A},$ there exists a $g_{\alpha}: [0, \infty) \to [0, 1]$ such that,
            \begin{align}
                & \alpha(x, y) = g_{\alpha}(s(x, y)), \ \  x, y \in \mathcal{X}, \\
                & g_{\alpha}(s) = sg_{\alpha}\left( \dfrac{1}{s}\right), \ \ 0 \le s < \infty, \label{eq:detailed_bal}\\
                & g_{\alpha}(e^x), x \in \mathbb{R} \text{ is Lipschitz continuous.} \label{eq:g_lipschitz}
            \end{align}
    \end{enumerate}

    Consider the class $\mathcal{G}$ of balancing functions defined on $[0, \infty)$ with the following properties,
    \begin{enumerate}
        \item For each $g \in \mathcal{G}$
    \end{enumerate}    

    Consider the class $\mathcal{G}$ of balancing functions such that $g \in \mathcal{G}$ \textit{if and only if} it satisfies the following properties,
    \begin{align}
        & g: [0, \infty) \to [0, 1],  \label{eq:gfunc}\\
        & g(s) = sg\left( \dfrac{1}{s}\right), \ \ 0 \le s < \infty, \label{eq:detailed_bal}\\
        & g(e^x), x \in \mathbb{R} \text{ is Lipschitz continuous.} \label{eq:g_lipschitz}
    \end{align}
    Define another class of functions $\mathcal{A}$ as,
    \begin{equation}
    \label{eq:A}
        \mathcal{A} := \{\alpha(x, y): \alpha(x, y) = g(s(x, y)) \text{ for some }g \in \mathcal{G} \}. 
    \end{equation}
\end{comment}

%%%%%%%

    Properties~\eqref{eq:gfunc} and \eqref{eq:detailed_bal} are standard and easy to verify, with \eqref{eq:gfunc} ensuring intractable constants in $\pi$ cancel away and \eqref{eq:detailed_bal} ensuring $\bs{\pi}$-reversibility.
     % \citep{billera2001}.
    % Then for any symmetric proposal density $q$, each $\alpha \in \mathcal{A}$ is a valid acceptance function, in the sense that it leads to a reversible Markov chain with respect to $\bs{\pi}$ .
    Property~\eqref{eq:g_lipschitz} is not required for $\alpha$ to be a valid acceptance function, however, we need it  for our optimal scaling results (to establish Lemma \ref{lemm:uni2}) and holds true for all common acceptance probabilities. Moreover, each $\alpha \in \mathcal{A}$ can be identified by the corresponding $g_{\alpha}$ and we will use $\alpha$ and $g_{\alpha}$ interchangeably.
    % . Hence, depending on the context, we will be using either of the two to refer to any particular acceptance function. 
    
    If $g_{\text{MH}}$ denotes the balancing function for MH acceptance function $(\alpha_{\text{MH}})$, then,
    \begin{equation}
    \label{eq:mh}
        g_{\text{MH}}(z) = 1 \wedge z, \ \ \ z \ge 0.
    \end{equation}
    It is easy to see that $\alpha_{\text{MH}} \in \mathcal{A}$. The Lazy-MH $(\alpha_{\text{L}})$ acceptance of \cite{latuszynski2013,herbei:berliner:2014} also belongs to $\mathcal{A}$. For a fixed $\epsilon \in [0, 1]$, it is defined using,
    \begin{equation}
    \label{eq:lazy}
        g_{\text{L}}(z) = (1 - \epsilon)(1 \wedge z),  \ \ \ z \ge 0\,.
    \end{equation} 
    % as its balancing function.

    The Barker's acceptance function is  $\alpha_{\text{B}}(x, y) = g_{\text{B}}(s(x, y))$ for all $x, y \in \X$ where,
    \begin{equation}
    \label{eq:barker}
        g_{\text{B}}(z) = \dfrac{z}{1+z}, \ \ \ z \ge 0.
    \end{equation}
    Then, \eqref{eq:detailed_bal}  follows immediately. For differentiable functions, property \eqref{eq:g_lipschitz}, i.e. Lipschitz continuity of $g_{\alpha}(e^z)$ can be verified by bounding the first derivative. In particular, we have $|g'_{\text{B}}(e^z)| \le 1$ for all $z \in \mathbb{R}$ and hence, $\alpha_{\text{B}} \in \mathcal{A}$. Due to \cite{peskun1973}, it is well known that in the context of Monte Carlo variability of ergodic averages, MH is superior to Barker's. Even so, the Barker's acceptance function has had a recent resurgence aided by its use in Bernoulli factory MCMC algorithms for Bayesian intractable posteriors where MH algorithms are not implementable. 

    We present a generalization of \eqref{eq:barker}; for $r \geq 1$ define
    \begin{equation*}
    g_r^{\text{R}}(z) = \begin{cases}
      \ds \frac{z(z^r - 1)}{z^{r+1} - 1}, & z \neq 1 \\
      \ds \frac{r}{r + 1}, & z  = 1\,.
    \end{cases}
  \end{equation*}
For $r \in \mathbb{N}$, the above can be rewritten as:
    \begin{equation}
    \label{eq:bark_seq}
        g_r^{\text{R}}(z) = \frac{z + \dots + z^r}{1 + z + \dots + z^r}, \ \quad z \ge 0, \,   r \in \mathbb{N}.
    \end{equation} 
    If $\alpha_r^{\text{R}}$ is the associated acceptance function, then, $\alpha_r^{\text{R}} \in \mathcal{A}$ for all $r \geq 1$. Moreover, $g_1^{\text{R}} \equiv g_{\text{B}}$ and $g_r^{\text{R}} \uparrow g_{\text{MH}}$ as $r \to \infty$. For $r \in \mathbb{N}$, we present a natural Bernoulli factory in the spirit of \cite{gonccalves2017b} that generates events of probability $\alpha_r^{\text{R}}$ without explicitly evaluating it; see Appendix~\ref{sec:BF}. An alternative approach would be to follow the general sampling algorithm of \cite{morina2019bernoulli} for rational functions.

    Let $\bs{\Phi}(\cdot)$ be the standard normal distribution function. For a theoretical exposition, \cite{bedard2008} defines the following acceptance probability for some $h > 0$:
    \begin{equation}
    \label{eq:bedard_seq}
        g_h^{\text{H}}(z) =  \bs{\Phi} \left( \frac{\log z - h/2}{\sqrt{h}}\right) + z \cdot \bs{\Phi} \left( \frac{-\log z - h/2}{\sqrt{h}}\right), \ \ \ z \ge 0. 
    \end{equation}
    For each $h > 0$, $\alpha_h^{\text{H}} \in \mathcal{A}$ and observe that as $h \to 0$, $g_h^{\text{H}} \to g_{\text{MH}}$ and as $h \to \infty, g_h^{\text{H}} \to 0$, i.e. the chain never moves. Similar examples can be constructed by considering other well behaved distribution functions in place of $\bs{\Phi}$. Lastly, it is easy to see that $\mathcal{A}$ is convex. Thus, it also includes situations when each update of the algorithm randomly chooses an acceptance probability. Moreover, as evidenced in \eqref{eq:lazy}, $\mathcal{A}$ is also closed under scalar multiplication as long as the resulting function lies in $[0, 1]$.

\section{Main theorem}  % (fold)
\label{sec:main}

    Let $f$ be a 1-dimensional density function and consider a sequence of target distributions $\{ \bs{\pi}_d \}$ such that for each $d$, the joint density is
    \begin{equation*}
    \label{eq:target}
        \pi_d(\bs{x}^d) = \prod_{i=1}^df(x_i^d), \ \ \ \ \ \bs{x}^d = (x_1^d, \dots, x_d^d)^T\in \mathbb{R}^d.
    \end{equation*}
    % That is, the components of the target distribution are independently and identically distributed. 
    
    \begin{assum}
        \label{assum:f}
    Density $f$ is positive and in $C^2$--the class of all real-valued functions with continuous second order derivatives. Further, $f'/f$ is Lipschitz and the following moment conditions hold,
    \begin{equation}
    \label{eq:A1A2}
        \mathbb{E}_f\left[ \left(\frac{f'(X)}{f(X)}\right)^8\right] < \infty, \hspace{20pt} \mathbb{E}_f\left[ \left(\frac{f''(X)}{f(X)}\right)^4\right] < \infty.            
    \end{equation}        
    \end{assum}

    Consider the sequence of Gaussian proposal kernels $\{ Q_d(\bs{x}^d, \cdot) \}$ with associated density sequence $\{q_d\}$, so that $Q_d(\bs{x}^d, \cdot) = N(\bs{x}^d, \sigma^2_d\bb{I}_d)$ where for some constant $l \in \mathbb{R}^{+}$,
    \[
        \sigma^2_d = l^2/(d-1)\,.
    \]
     %
    % such that,
    % \begin{equation}
    % \label{eq:prop}
    %     q_d(\bs{x}^d, \bs{y}^d) = \dfrac{1}{(2\pi\sigma^2_d)^{d/2}}\exp\left\{\dfrac{-1}{2\sigma^2_d}\mid \bs{y}^d - \bs{x}^d \mid^2\right\}, \ \ \bs{x}^d, \bs{y}^d \in \mathbb{R}^d.
    % \end{equation}
    The proposal $Q_d$ is used to generate a $d-$dimensional Markov chain, $\bs{X}^d = \{\bs{X}^d_n, n \ge 0\}$, following the accept-reject mechanism with acceptance function $\alpha$.
     % Same procedure is repeated for each $d \in \mathbb{N}$ to obtain a sequence of Markov chains $ \{\bs{X}^d, d \ge 1\} $. 
    Under these conditions and with  $\alpha = \alpha_{\text{MH}}$, \cite{roberts1997weak} established weak convergence to an appropriate Langevin diffusion for the sequence of 1-dimensional stochastic processes, constructed from the first component of these Markov chains. Since the coordinates are independent and identically distributed, this limit informs the limiting behaviour of the full Markov chain in high-dimensions. In what follows, we extend their results to the class of acceptance functions, $\mathcal{A} $, as defined in Definition~\ref{assum:acceptance}.

% Thus, to avoid a trivial limiting process as dimension increases to infinity, the process must be sped up by a factor of $d$. So, 
Let $\{\bs{Z}^d, d > 1\}$ be a sequence of processes constructed by speeding up the Markov chains by a factor of $d$ as follows,
    \begin{equation*}
    \label{eq:zt}
        \bs{Z}^d_t = \bs{X}^d_{[dt]} = (X^d_{[dt],1}, X^d_{[dt],2}, \dots, X^d_{[dt],d})^T; \ \ \ t > 0.
    \end{equation*}
    Suppose $\{\eta_d: \mathbb{R}^d \to \mathbb{R}\}$ is a sequence of projection maps such that $\eta_d(\bs{x}^d) = x_1^d$. Define a new sequence of $1$-dimensional processes $\{U^d, d > 1\}$ as follows,
    \begin{equation*}
    \label{eq:Ut}
        U^d_t := \eta_d \circ \bs{Z}^d_t = {X}^d_{[dt], 1}; \ \ \ t > 0.
    \end{equation*}
\begin{comment}
    The first component of each process $\bs{Z}^d$ is then extracted to obtain a new sequence of $1$-dimensional processes $\{U^d, d > 1\}$ such that,
    \begin{equation*}
    \label{eq:Ut}
        U^d_t = \bs{Z}^d_{t, 1} = {X}^d_{[dt], 1}; \ \ \ t > 0.
    \end{equation*}
\end{comment}
    Under stationarity, we show that $\{U^d, d > 1\}$ weakly converges  \cite[in the Skorokhod topology, see][]{ethier2009} to a Markovian limit $U$. We denote weak convergence of processes in the Skorokhod topology by ``$\Rightarrow$" and standard Brownian motion at time $t$ by $B_t$. The proofs are in the appendices. 
    \begin{theorem}
    \label{thm:main}
        % Let $f$, positive and $C^2$, be such that $f'/f$ is Lipschitz and $f$ satisfies the regularity conditions \eqref{eq:A1A2}. Consider the sequence of target densities $\{ \pi_{d} \}$ where each $\pi_d$ has the product form as in (\ref{eq:target}). 

        Let $\{\bs{X}^d, d \ge 1\}$ be the sequence of $\bs{\pi}_d$-invariant Markov chains constructed using acceptance function $\alpha$ and proposal $Q_d$ such that $\bs{X}^d_0 \sim \bs{\pi}_d$. Further, suppose $\alpha \in \mathcal{A}$ %satisfies Assumption~\ref{assum:acceptance}%
        and $\bs{\pi}_d$ satisfies Assumption~\ref{assum:f}. 
        % Let the proposal variance be scaled so that $\sigma^2_d = l^2/(d-1)$ for some constant $l \in \mathbb{R}$.
        % 
       % Construct a sequence of one-dimensional process $\{U^d, d > 1\}$ from the sequence $\{\bs{X}^d, d \ge 1\}$ as per the relation (\ref{eq:Ut}). Finally, let $\bs{X}^d_0 \sim \pi$ for all $d$.
% 
        Then, $U^d \Rightarrow U$, where $U$ is a diffusion process that satisfies the Langevin stochastic differential equation,
        \begin{equation*}
            dU_t = (h_{\alpha}(l))^{1/2}dB_t + h_{\alpha}(l)\frac{f'(U_t)}{2f(U_t)}dt,
        \end{equation*}
        with $h_{\alpha}(l) = l^2 M_{\alpha}(l)$, where, 
        \begin{equation}
            \label{eq:mal}
            M_{\alpha}(l) = \int_{\mathbb{R}} g_{\alpha}(e^b) \frac{1}{\sqrt{2\pi l^2I}} \exp\left \{ \frac{-(b + l^2I/2)^2}{2l^2I}\right\} db,
        \end{equation}
        and, 
    \[
        I = \mathbb{E}_{f}\left[ \left(\frac{f'(X)}{f(X)}\right)^2 \right].
    \]
    \end{theorem}

    \begin{remark}
    Since $\alpha_{\text{MH}} \in \mathcal{A}$, our result aligns with \cite{roberts1997weak} since
    \[
        M_{\text{MH}} (l) = \int_{\mathbb{R}} g_{\text{MH}}(e^b) \frac{1}{\sqrt{2\pi l^2I}} \exp\left \{ \frac{-(b + l^2I/2)^2}{2l^2I}\right\} db = 2 \Phi\left(-\frac{l\sqrt{I}}{2} \right)\,.
    \]
     % Hence, the result for MH acceptance function becomes a special case of this theorem.
    \end{remark}
    \begin{remark}
        For symmetric proposals, Definition~\ref{assum:acceptance} requires $\alpha$ to be a function of only the ratio of the target densities at the two contested points. Thus, the result is not applicable to acceptances in \cite{mira2001,banterle2019accelerating, vats2020}. 
    \end{remark}

    In Theorem~\ref{thm:main}, $h_{\alpha}(l)$ is the speed measure of the limiting diffusion process and so the optimal choice of $l$ is $l^*$ such that
    % the one that  maximizes $h_{\alpha}(l)$. Let this optimal $l$ be denoted by $l^*$. So,
    \[
        l^* = \underset{l}{\mathrm{arg\,max}}\, h_{\alpha}(l).
\]
       Denote the average acceptance probability by
    \[
        \alpha_{d}(l) := \mathbb{E}_{\bs{\pi}_d, Q_d} \left[ \alpha(\bs{X}^d, \bs{Y}^d) \right] = \int \int \pi(\bs{x}^d) \ \alpha(\bs{x}^d, \bs{y}^d) \ q_d(\bs{x}^d, \bs{y}^d) \ d\bs{x}^d \ d\bs{y}^d\,,
    \]
    and the asymptotic acceptance probability as $\alpha(l) := \lim_{d \to \infty} \alpha_{d}(l)$. The dependence on $l$ is through the variance of proposal kernel.  We then have the following corollary.
    \begin{corollary} \label{corr:1} Under the setting of Theorem \ref{thm:main}, we obtain $\alpha(l) = M_{\alpha}(l)$ and the asymptotically optimal acceptance probability is  $M_{\alpha}(l^*)$.
    \end{corollary}
    
    Corollary~\ref{corr:1} is of considerable practical relevance since for different acceptance functions it yields the optimal target acceptance probability to tune to.

    \subsection{Optimal results for some acceptance functions} % (fold)
    \label{sec:barker}

        In Section~\ref{sec:acc_func}, we discussed some important members of the class $\mathcal{A}$. Corollary~\ref{corr:1} can then be used to obtain the AOAR for them by maximizing the speed measure of the limiting diffusion process. For Barker's algorithm, from Theorem \ref{thm:main} and \eqref{eq:barker}, the speed measure $h_{\text{B}}(l)$ of the corresponding limiting process is $h_{\text{B}}(l) = l^2M_{\text{B}}(l)$ where,
        \[
            M_{\text{B}}(l) = \int_{\mathbb{R}} \frac{{}1}{1 + e^{-b}} \frac{1}{\sqrt{2\pi l^2I}} \exp\left \{ \frac{-(b + l^2I/2)^2}{2l^2I}\right\} db.
        \]
        Maximizing $h_{\text{B}}(l)$, the optimal value, $l^*$, is approximately (see Appendix \ref{sec:speed}),
        \[
            l^* =  \frac{2.46}{\sqrt{I}}\,.
        \] 

        \begin{figure}[htbp]
            \centering
            \includegraphics[height = 1.5in]{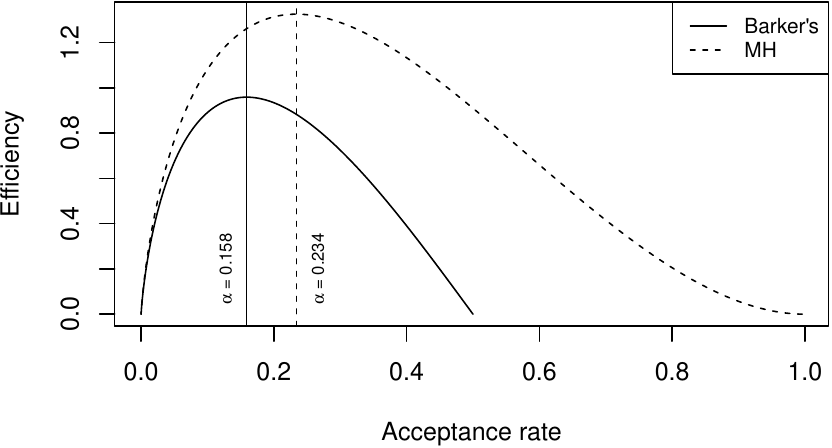}  \hspace{.5cm}
            \includegraphics[height = 1.9in]{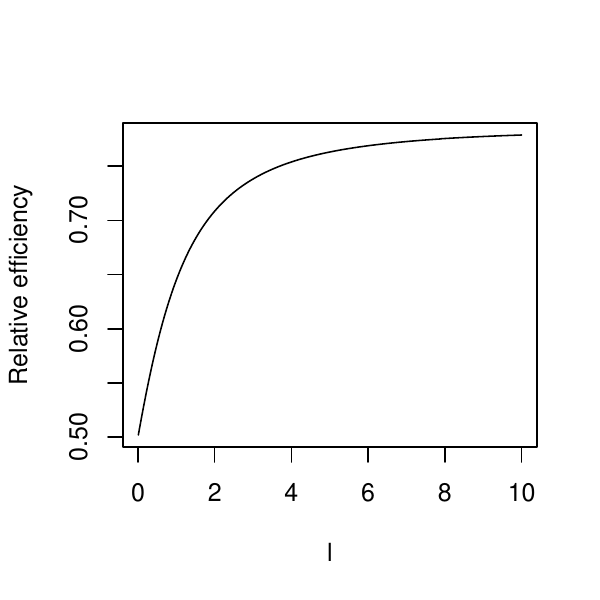} 
            \caption{Efficiency $(h(l))$ versus acceptance rate $(\alpha(l))$ with $I = 1$ (left). Relative efficiency of Barker's vs MH $(h_{\text{B}}(l)/h_{\text{MH}}(l) )$ versus $l$ (right).}
            \label{fig:eff_acc}
        \end{figure}

\begin{comment}
        By Corollary~\ref{corr:1}, using this $l^*$ yields an asymptotic acceptance rate of approximately $0.158$. %This leads to the following useful heuristic for people working with Barker's acceptance function. 
        Hence, when the optimal variance is not analytically tractable in high dimensions, one may consider tuning their algorithm so as to achieve an acceptance probability of approximately $0.158$. Additionally, the ratio of the speed measures of Barker's versus MH at their respective optimal scalings is $0.72$, which gives a direct quantification of the result in \cite{latuszynski2013}. We can also study the respective speed measures as a function of the acceptance rate; this is given in the left plot in Figure~\ref{fig:eff_acc}. We find that as the AOAR increases, the speed measure for Barker's decreases more rapidly than MH. The right plot of Figure~\ref{fig:eff_acc} also indicates that the relative efficiency of Barker's versus MH, for a fixed $l$.
\end{comment}

        By Corollary~\ref{corr:1}, using this $l^*$ yields an asymptotic acceptance rate of approximately $0.158$. %This leads to the following useful heuristic for people working with Barker's acceptance function. 
        Hence, when the optimal variance is not analytically tractable in high dimensions, one may consider tuning their algorithm so as to achieve an acceptance probability of approximately $0.158$. Additionally, the right plot in Figure \ref{fig:eff_acc} verifies that the relative efficiency of Barker's versus MH, as measured by the ratio of their respective speed measures for a fixed $l$, remains above $0.5$ \citep[see Theorem 4 in ][]{latuszynski2013}; this relative efficiency increases as $l$ increases. Additionally, the ratio of the speed measures of Barker's versus MH at their respective optimal scalings is $0.72$. This quantifies the loss in efficiency in running the best version of Barker's compared to the best version of MH algorithm. We can also study the respective speed measures as a function of the acceptance rate; this is given in the left plot in Figure~\ref{fig:eff_acc}. We find that as the asymptotic acceptance rate increases, the speed measure for Barker's decreases more rapidly than MH. This suggests that there is much to gain by appropriately tuning the Barker's algorithm.

        \begin{figure}[htbp]
            % \vspace{-20pt}
            \centering
            \includegraphics[width = 2.3in]{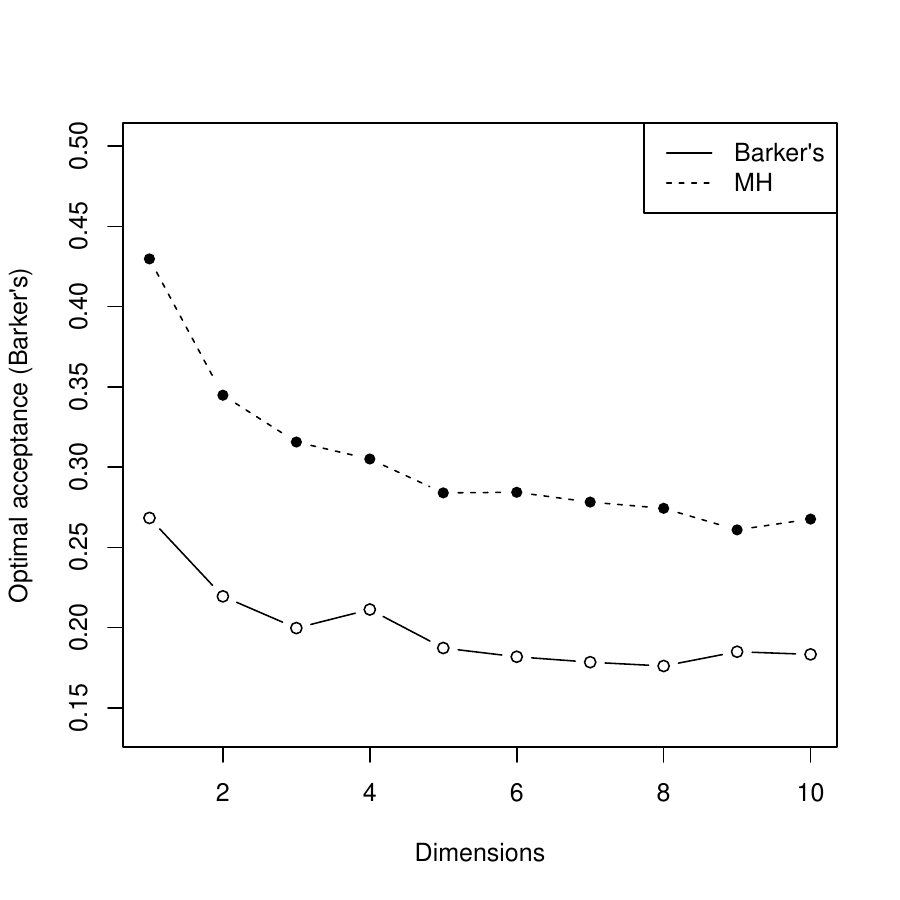}
            \caption{Optimal acceptance rate against number of dimensions.}
            \label{fig:aoar}
        \end{figure}

        For lower dimensions, the optimal acceptance rate is higher than the AOAR. Figure \ref{fig:aoar} shows optimal values for MH and Barker's algorithms on isotropic Gaussian targets in dimensions $1$ to $10$; proposal kernel being the same as in the setting of Theorem \ref{thm:main}. This plot is produced using the criterion of minimizing first order auto-correlations in each component \citep{gelman1996, roberts1998,roberts2001}. For $\alpha_{\text{MH}}$ and $\alpha_{\text{B}}$, the optimal acceptance rate in one dimension is $0.43$ and $0.27$ respectively.
        % , which is significantly higher than their respective asymptotically optimal acceptance rates. 
        % However, the convergence kicks in quite early and even in $5$ dimensions, the optimal acceptance rate is so close to the asymptotic one so as to not matter at all in practice.

        %\vspace{-10pt}

        For Lazy-MH with $\epsilon \in [0, 1]$, Corollary \ref{corr:1}  implies that the AOAR of the algorithm is $(1 - \epsilon)0.234$ with the same optimal $l^*$ as MH.  For the acceptance functions, $\alpha_h^{\text{H}}$ in \eqref{eq:bedard_seq},
        \[
            M_h(l) = 2\bs{\Phi} \left(-\frac{\sqrt{h + l^2I}}{2} \right)\,.
        \]
        %\begin{align*}
            %M_h(l) &= \int_{\mathbb{R}} \left(\bs{\Phi} \left( \frac{b - h/2}{\sqrt{h}}\right) + e^b \bs{\Phi} \left( \frac{- b - h/2}{\sqrt{h}}\right)\right) \frac{1}{\sqrt{2\pi l^2I}} \exp\left \{ \frac{-(b + l^2I/2)^2}{2l^2I}\right\} db, \\
            %&= 2\bs{\Phi} \left(-\frac{\sqrt{h + l^2I}}{2} \right)
        %\end{align*}
        With $h = 0$, we obtain the result of \cite{roberts1997weak} for MH. Further, the left plot of Figure \ref{fig:optimal} highlights that as $h \to 0$, the AOAR increases to $0.234$ and the algorithm worsens as $h$ increases. Moreover, for $h \approx 1.913$, the AOAR is roughly $0.158$, i.e. equivalent to the Barker's acceptance function. 
        % At this value of $h$, $\alpha_h^{\text{H}}$ is almost equal to $\alpha_{\text{B}}$ with slight positive difference when the ratio of target densities become large. This difference is balanced by $\alpha_{\text{B}}$ for small values of the ratio $s(x, y)$.

        \begin{figure}[htbp]
            \centering
            \includegraphics[height = 2.2in]{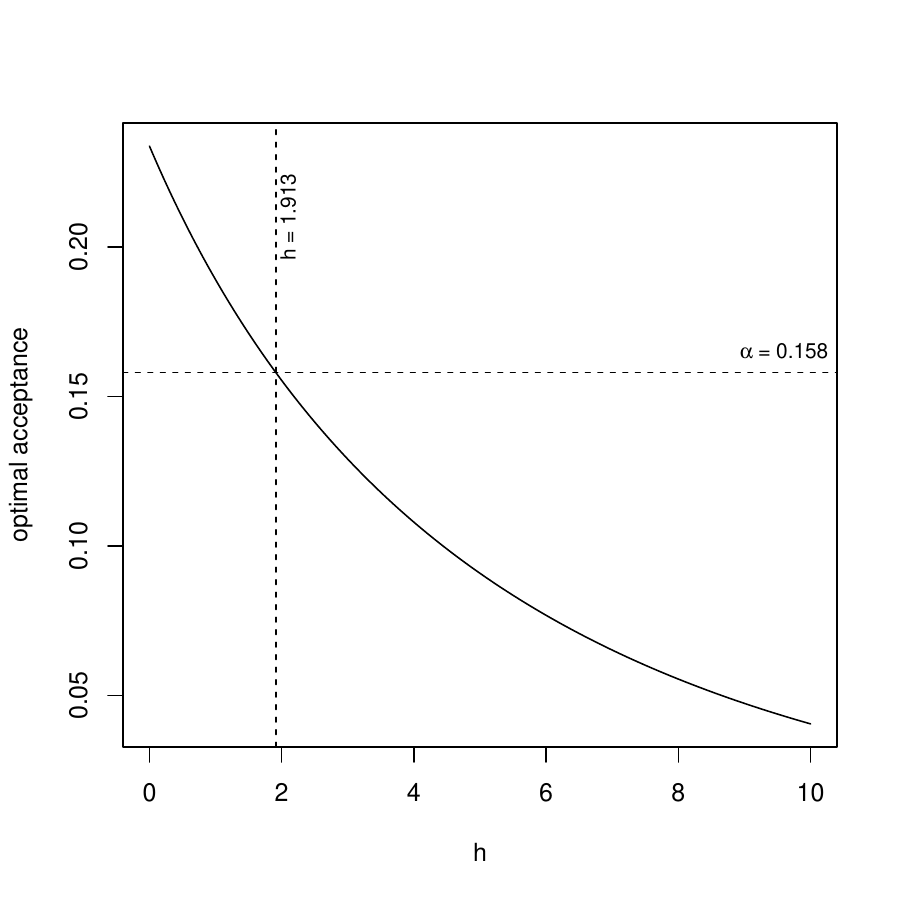}  \hspace{.5cm}
            \includegraphics[height = 2.2in]{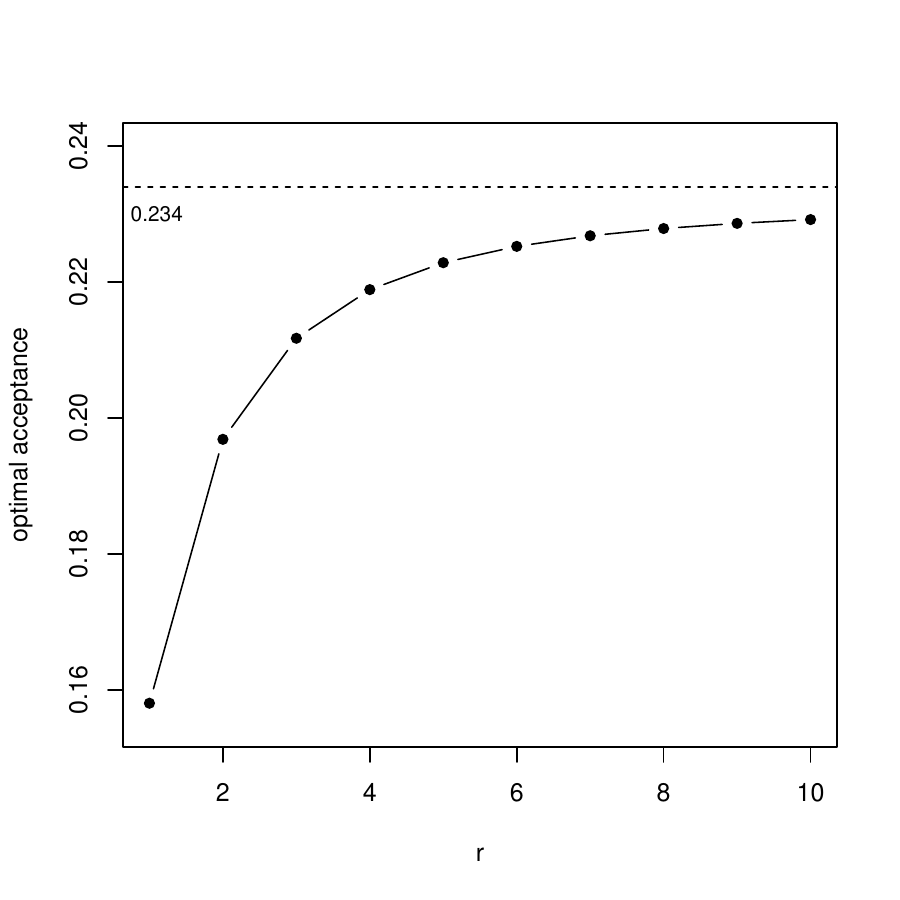} 
            \caption{Optimal acceptance rates for $\alpha^{\text{H}}_h$ against $h$ (left) and $\alpha^{\text{R}}_r$ against $r$ (right). }
            \label{fig:optimal}
        \end{figure}

\begin{comment}
        \begin{figure}[htbp]
            % \vspace{-30pt}
            % \flushleft
            % \hspace{-10pt}
            \begin{subfigure}{0.5\textwidth}
                \includegraphics[width = 2.2in]{bedard.pdf} 
                \caption{$\alpha^{\text{H}}_h$: For different values of $h$.}
                \label{fig:bedard}
            \end{subfigure}
            \begin{subfigure}{0.5\textwidth}
                \flushright
                \includegraphics[width = 2.2in]{barker.pdf} 
                \caption{$\alpha^{\text{R}}_r$: For different values of $r$.}
                \label{fig:barker}
            \end{subfigure}
            % \vspace{-.2cm}
            \caption{Optimal acceptance rates.}
            \label{fig:optimal}
        \end{figure}
\end{comment}        

        Lastly, the AOARs for $\alpha_r^{\text{R}}$ in \eqref{eq:bark_seq} are available. For $r = 1, \dots, 10$, the results have been plotted in the right plot of Figure \ref{fig:optimal}. As anticipated, the AOAR approaches $0.234$ as $r$ increases. Notice that $\alpha_2^{\text{R}}$ yields an AOAR of $0.197$, which is a considerable increase from $\alpha_B = \alpha_1^{\text{R}}$.  
         % Thus, our generalized Barker's Bernoulli factory, when efficiently implementable, can be tuned to this rate.  
         Table \ref{tab:results} below summarizes the results of this section.$^{1}$\footnote{$^1$ Codes for all  plots and tables are available at \texttt{https://github.com/Sanket-Ag/BarkerScaling}}
        \begin{table}[htbp]
            \centering
            \begin{tabular}{|c|c|ccc|ccc|c|}
                \hline
                 & $\alpha_{\text{MH}}$ & $\alpha^{\text{H}}_1$ & $\alpha^{\text{H}}_{1.913}$ & $\alpha^{\text{H}}_5$ & $\alpha^{\text{R}}_{10}$ & $\alpha^{\text{R}}_5$ & $\alpha^{\text{R}}_2$ & $\alpha_{\text{B}}$ \\
                \hline
                $M_{\alpha}(l^*)$ & 0.234 & 0.189 & 0.158 & 0.129 & 0.229 & 0.223 & 0.197 & 0.158 \\
                \hline
                $|l^*\sqrt{I}|$ & 2.38 & 2.43 & 2.46 & 2.49 & 2.39 & 2.39 & 2.42 & 2.46 \\
                \hline
            \end{tabular}
            \caption{Optimal proposal variance and asymptotic acceptance rates.}
            \label{tab:results}
        \end{table}

\section{Numerical results}
\label{sec:sims}
    We study the estimation quality for different expectations as a function of the proposal variance (acceptance rate) for the generalized Barker's acceptance function, $\alpha_r^{\text{R}}$. We focus on $r = 1$ (Barker's algorithm) and $r = 2$.
    % Theorem \ref{thm:main} is practical in running efficient Markov chain when the function depends on a single component. However, theoretical results obtained in Section \ref{sec:barker} for optimal acceptance rates are robust to the function of interest and extend to more complicated target distributions than those assumed in Theorem \ref{thm:main}. \cite{roberts2001} verified this analytically for $\alpha_{\text{MH}}$. We present simulation results to see this empirically for $\alpha_{\text{B}}$. 
    Suppose $f:\mathbb{R}^d \to \mathbb{R}$ is the function whose expectation with respect to $\bs{\pi}_d$ is of interest. Let $\{f(\bs{X}_n)\}$ be the mapped process. Similar to \cite{roberts2001}, we assess choice of proposal variance by the convergence time:
    \begin{equation*}
        \text{convergence time } := \frac{-k}{\log (\rho_k)}\,,
    \end{equation*}
    where $\rho_k$ is the lag-$k$ autocorrelation in $\{f(\bs{X}_n)\}$. In each of the following simulations, convergence time is estimated by averaging over $10^3$ replications of Markov chains, each of length $10^6$ with $k = 1$. We chose a range of values of $l$ where $l$ is such that $\sigma^2_d = l^2/d$ in a Gaussian proposal kernel $Q_d(\bs{x}^d, \cdot) = N(\bs{x}^d, \sigma^2_d \bb{I}_d)$.
    
    Consider first the case of an isotropic target, $\bs{\pi}_{d} = N_{d}(\bs{0}, \bb{I}_{d})$ with isotropic Gaussian proposals; the conditions of Theorem \ref{thm:main} are satisfied. The estimated convergence time for $f(\bs{x}) = x_1$ and $f(\bs{x}) = \bar{\bs{x}}$ where $\bar{\bs{x}}$ is the mean of all components, $x_1, \dots, x_d$, is plotted in Figure \ref{fig:ideal} (top row). Here, $d = 50$. 
    \begin{figure}[htbp]
        \centering
        % \begin{subfigure}[b]{0.3\textwidth}
        %     \centering
        %     \includegraphics[width=\textwidth]{figures/ideal/muti_gaussian_1-1.pdf}
        % \end{subfigure}
        % \hfill
        \begin{subfigure}[b]{0.45\textwidth}
            \centering
            \includegraphics[width=2.3in]{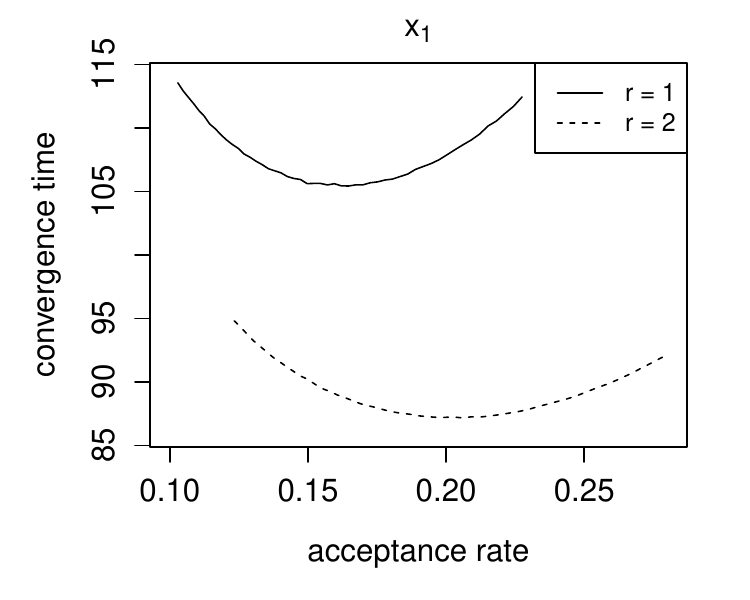}
        \end{subfigure}
        \begin{subfigure}[b]{0.45\textwidth}
            \centering
            \includegraphics[width=2.3in]{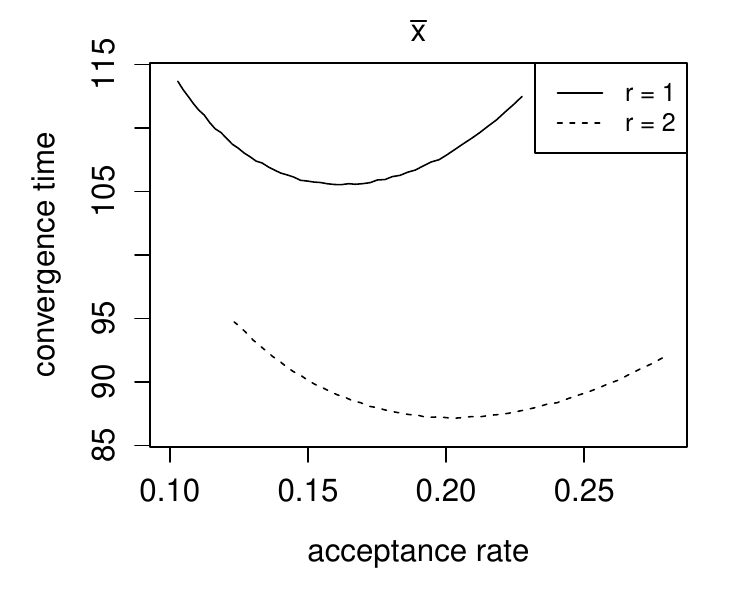}
        \end{subfigure} \\ 
        %         \begin{subfigure}[b]{0.3\textwidth}
        %     \centering
        %     \includegraphics[width=\textwidth]{figures/indProp/25Aug_rho085_1-1.pdf}
        % \end{subfigure}
        % \hfill
        \begin{subfigure}[b]{0.45\textwidth}
            \centering
            \includegraphics[width=2.3in]{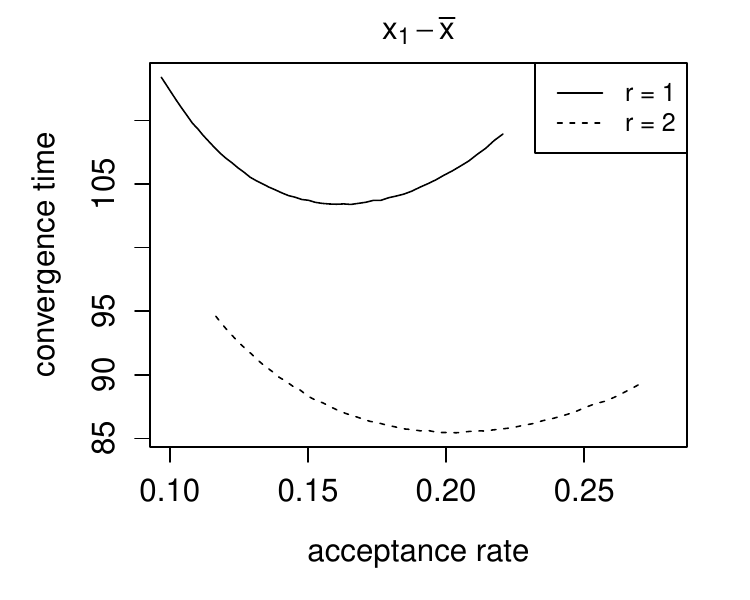}
        \end{subfigure}
        \begin{subfigure}[b]{0.45\textwidth}
            \centering
            \includegraphics[width=2.3in]{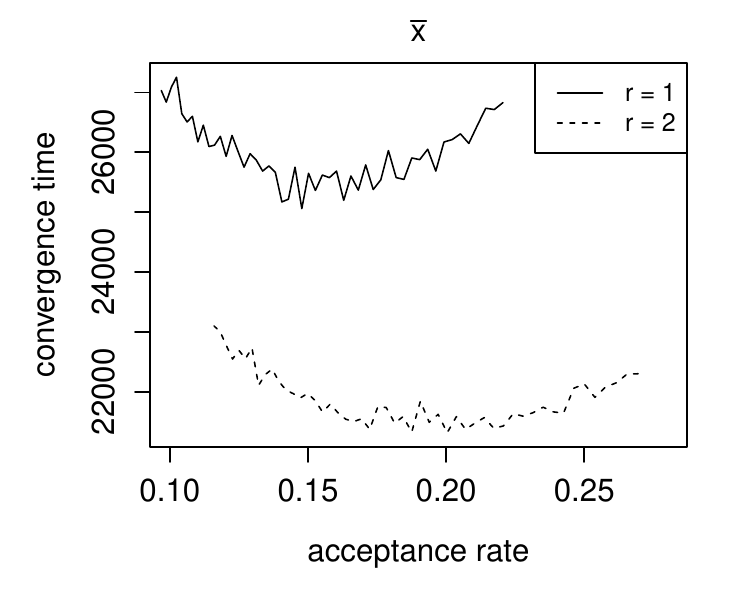}
        \end{subfigure}
        \caption{ Convergence times for $\alpha_{\text{B}}$ against acceptance rate in the isotropic setting (top row) and the correlated target setting (bottom row).}
        \label{fig:ideal}
    \end{figure}
    % The first plot is the estimated acceptance rate versus $\sigma = l/\sqrt{d}$ over a range of values for $l$. The horizontal line in the plot indicates the asymptotically optimal acceptance rate of $0.158$, which corresponds to a slightly larger value of $l$ than the theorized optimal value of $2.46$ (this is expected for finite dimensions). 
    % The middle plot is the convergence time for $x_1$ and the final plot is the convergence time for $\bar{\bs{x}}$. 
    For both functions of interest, the optimal performance i.e. the minimum convergence time, corresponds to an acceptance rate of approximately $0.158$ for $\alpha_{\text{B}}$ and $0.197$ for $\alpha_2^{\text{R}}$; the slight overestimation is due to the finite dimensional setting.
    
    Next, we consider $\bs{\pi}_d = N_d(\bs{0}, \bs{\Sigma}_d)$ where $\bs{\Sigma}_d$ is a $d \times d$ matrix with $1$ on its diagonal and all other elements are equal to some non-zero $\rho$. Here, the assumptions in Theorem~\ref{thm:main} are not satisfied. For such a target and for $\alpha_{\text{MH}}$, \cite{roberts2001} showed that the rate of convergence of the algorithm is governed by the eigenvalues of $\bs{\Sigma}_d$. In particular, the eigenvalues of $\bs{\Sigma}_d$ are $dp + 1 - \rho$ and $1 - \rho$ with associated eigenvectors $\bar{\bs{x}}$ and $x_i - \bar{\bs{x}} \, (i = 1,\dots,d)$, respectively. Then, it was shown that the algorithm converges quickly for functions orthogonal to $\bar{\bs{x}}$, but much more slowly for $\bar{\bs{x}}$. Despite the differing rates of convergence, the optimal acceptance rate, corresponding to the minimum convergence time, remains the same. We find this to be also true for $\alpha_{\text{B}}$ and $\alpha^{\text{R}}_2$ as illustrated in Figure \ref{fig:ideal} (bottom row) where we present convergence times for $x_1 - \bar{\bs{x}}$ and $\bar{\bs{x}}$. 
    % \begin{figure}
    %     \centering
    %     \begin{subfigure}[b]{0.3\textwidth}
    %         \centering
    %         \includegraphics[width=\textwidth]{figures/indProp/25Aug_rho085_1-1.pdf}
    %     \end{subfigure}
    %     \hfill
    %     \begin{subfigure}[b]{0.3\textwidth}
    %         \centering
    %         \includegraphics[width=\textwidth]{figures/indProp/25Aug_rho085_2-2.pdf}
    %     \end{subfigure}
    %     \hfill
    %     \begin{subfigure}[b]{0.3\textwidth}
    %         \centering
    %         \includegraphics[width=\textwidth]{figures/indProp/25Aug_rho085_3-end.pdf}
    %     \end{subfigure}
    %     \caption{ Correlated target with iid proposal S2. (Left) Acceptance rate of the Barker's algorithm against $\sigma$. (Middle and Right) Convergence times of the ordinary Barker's algorithm on multivariate Gaussian with dependent components against acceptance rate for functionals $x_1 - \bar{\bs{x}}$ and $\bar{\bs{x}},$ respectively. Here, $\rho = 0.85$ and $d = 50$}
    %     \label{fig:indProp}
    % \end{figure}
    % The first plot is again acceptance rate against $\sigma$. The value of $\sigma$ corresponding to the optimal acceptance rate of $0.158$ falls from approximately $0.357$ to $0.137$. This is natural since the target now has a correlated structure which restricts the algorithm from making larger jumps. 
    % The middle and the left plots are the convergence times for $x_1 - \bar{\bs{x}}$ and $\bar{\bs{x}}$ respectively. 
    Once again, $d = 50$. The large difference between convergence times for both is quite evident from the $y-$axis of the two plots. The minimum again lies in a region around the asymptotic optimal. We note that due to the slow convergence rate of $\bar{\bs{x}}$, the process demonstrates slow mixing, yielding more variable estimates of the convergence time. For both simulation settings, we see the expected improvement in the convergence time for $\alpha_2^{\text{R}}$ compared to $\alpha_{\text{B}}$.

    \subsection{A Bayesian logistic regression example}
        We consider fitting a Bayesian logistic regression model to the famous Titanic dataset which contains information on crew and passengers aboard the 1912 RMS Titanic ship. Let $\bs{y}$ denote the response vector (whether they survived or not) and $\bs{X}$ denote the $n \times d$ model matrix; here $d = 10$.  We assume a multivariate zero-mean Gaussian prior on $\bs{\beta}$ with covariance $100\bb{I}_{10}$. The resulting target density is
        \begin{equation*}
            \pi(\bs{\beta} \mid \bs{y}) \propto \exp \left\{ - \frac{\bs{\beta}^T\bs{\beta}}{2} \prod_{i = 1}^n \frac{\exp(-\bs{x}_i^T \bs{\beta})^{1 - y_i}}{1 + \exp(-\bs{x}_i^T \bs{\beta})}\right\}\,.
        \end{equation*}
        For the Titanic dataset, the resulting posterior has a complicated covariance structure with many components exhibiting an absolute mutual correlation of beyond .50. The posterior is also ill-conditioned with the condition number of the estimated target covariance matrix being $\approx 10^5$.        As seen in the bottom row of Figure \ref{fig:ideal}, in such situations an isotropic proposal kernel might perform poorly for most functions. We instead  consider a Gaussian proposal scheme where the proposal covariance matrix is taken to be proportional to the target covariance matrix. This is a common strategy for dealing with targets with correlated components and forms the basis for many adaptive MCMC kernels \citep{roberts2009examples}.         We implement the Barker's algorithm to sample from the posterior. Let $\bs{\Sigma}_d$ denote the covariance matrix associated with the posterior distribution of $\bs{\beta}$, then the proposal kernel $Q_d(\bs{x}^d, \cdot) = N(\bs{x}^d, \sigma^2_d \bs{\Sigma}_d)$. Since $\bs{\Sigma}_d$ is unavailable, we  estimate it from a pilot MCMC run of size $10^7$. We then consider various values of $\sigma^2_d = l^2/d$.

        The performance of the algorithm for different functions of interest is plotted in Figure~\ref{fig:bayes}. Since this is a 10-dimensional problem, the optimal acceptance rate from Figure~\ref{fig:aoar} is approximately $0.18$. The convergence times for both, $\beta_1 - \bar{\beta}$ and $\bar{\bs{\beta}}$, are similar. Further, both are minimized at approximately the same acceptance rate of $0.18$. It is natural here to be interested in estimating the posterior mean vector. Thus, we also study the properties of vector $\bs{\beta}$ with efficiency measured via the multivariate effective sample size (ESS) \citep{vats2019}. The ESS returns the  equivalent number of iid samples from $\bs{\pi}$ that would yield the same variability in estimating the posterior mean as the given set of MCMC samples. In Figure~\ref{fig:bayes}, we see that the optimal acceptance rate corresponding to the highest ESS values is achieved around $0.18$.
        \begin{figure}[htbp]
            \centering
 
            \begin{subfigure}[b]{0.32\textwidth}
                \centering
                \includegraphics[width=1.7in]{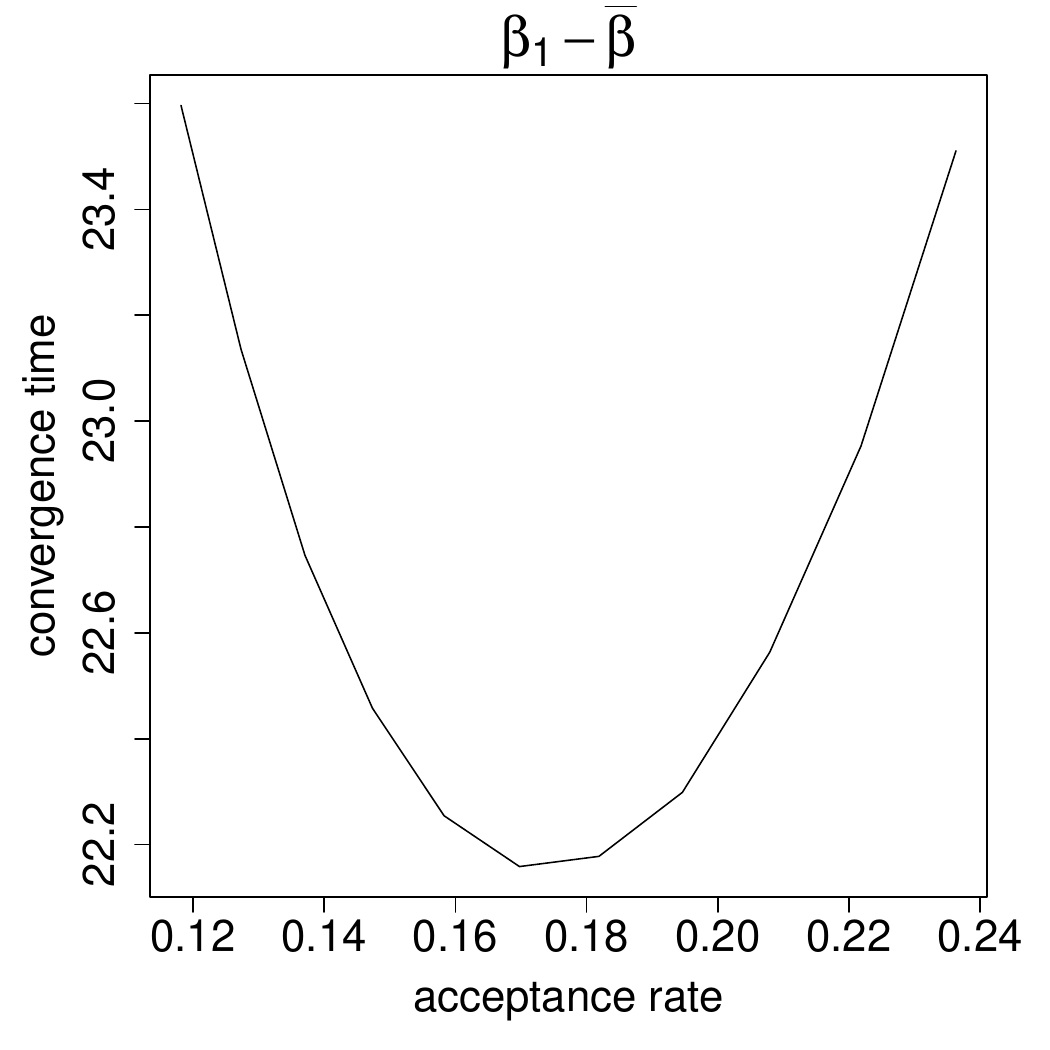}
            \end{subfigure}
            \begin{subfigure}[b]{0.32\textwidth}
                \centering
                \includegraphics[width=1.7in]{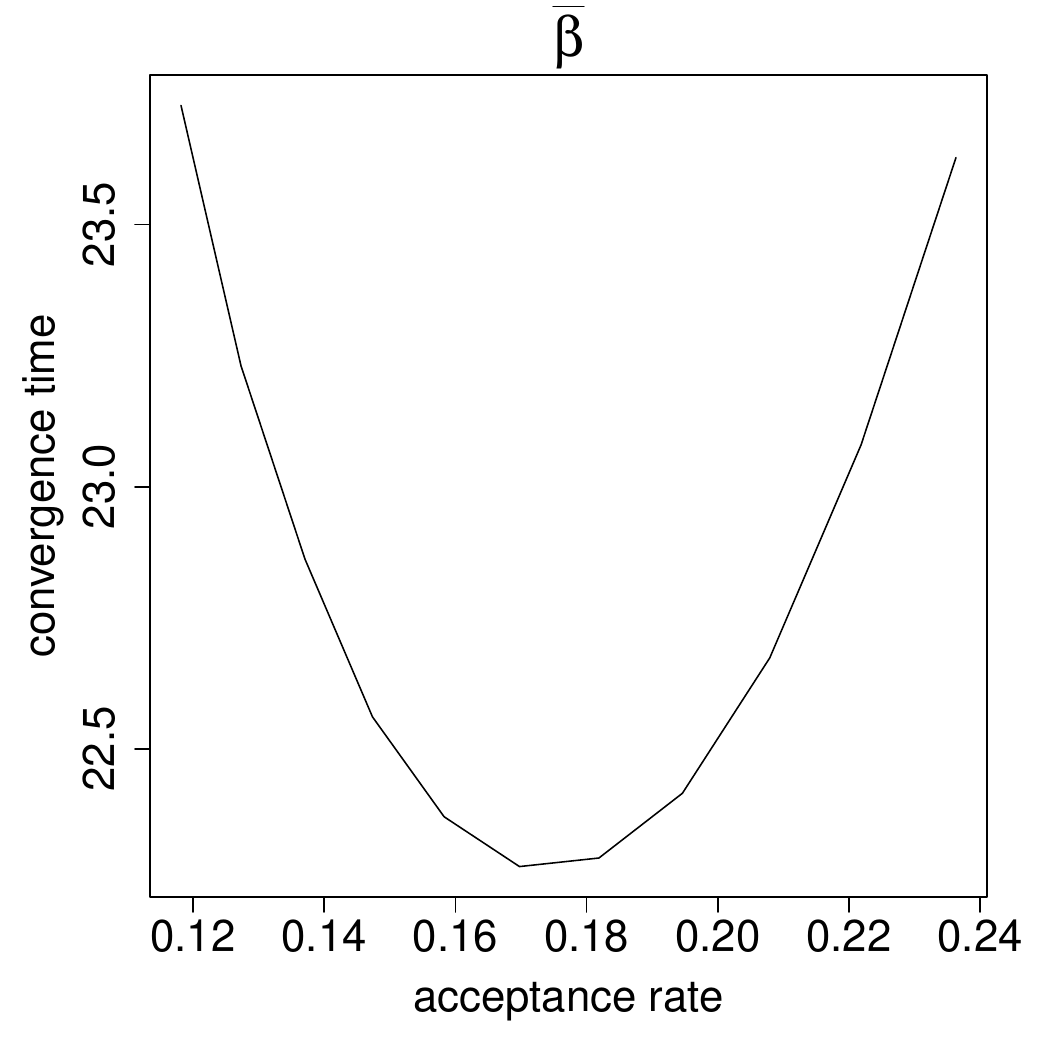}
            \end{subfigure}
            \begin{subfigure}[b]{0.32\textwidth}
                \centering
                \includegraphics[width=1.7in]{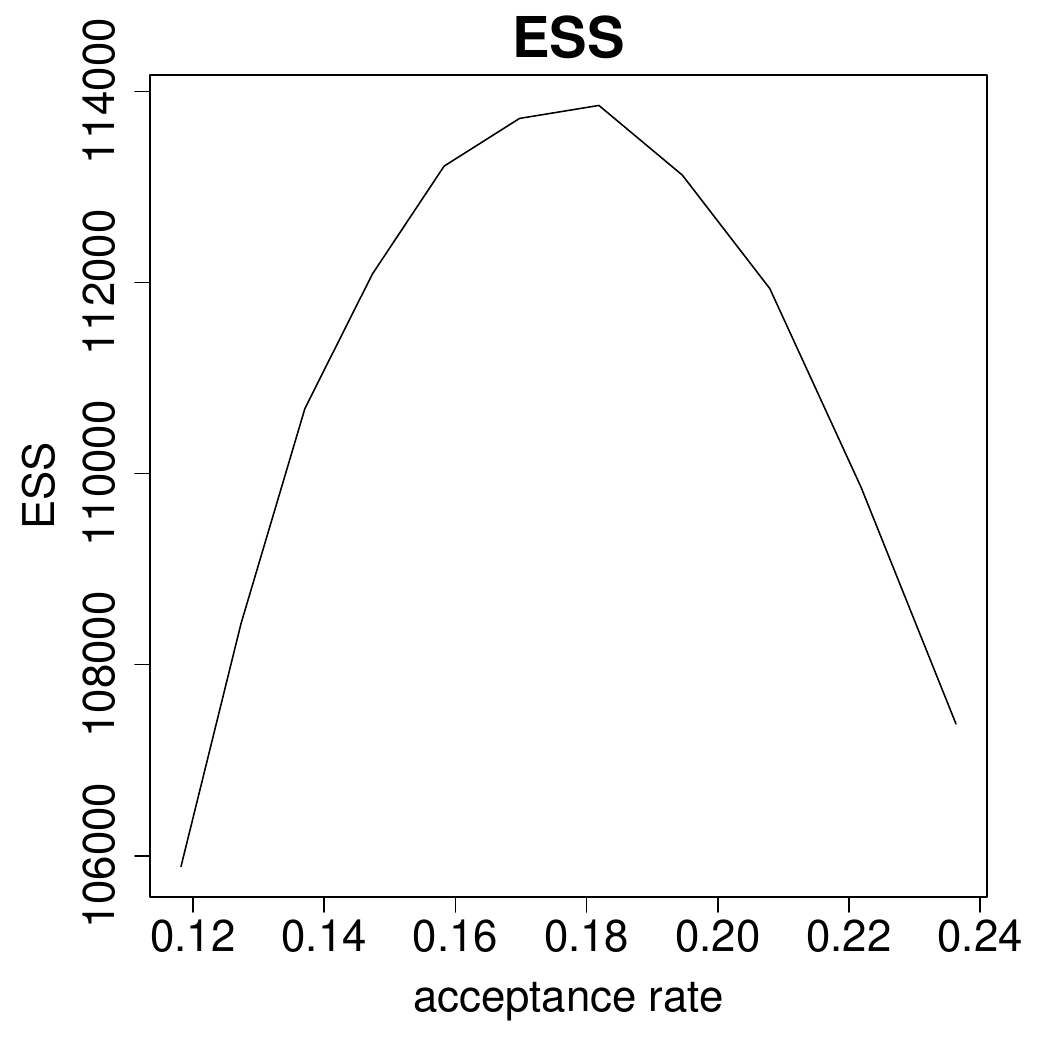}
            \end{subfigure}
            \caption{Convergence times for $\alpha_{\text{B}}$ (left and middle) and multivariate ESS for the posterior mean vector (right) against acceptance rate.}
            \label{fig:bayes}
        \end{figure}

\section{Conclusions}  % (fold)
\label{sec:conc}
    We obtain optimal scaling and acceptance rates for a large class of acceptance functions. In doing so, we found that the scaling factor of $1/d$ for the proposal variance holds for all acceptance functions, indicating that the acceptance functions are not likely to affect the rate of convergence, just the constants associated with that rate. Thus, practitioners need not hesitate in switching to other acceptance functions when the MH acceptance probability is not tractable, as long as Corollary~\ref{corr:1} is used to tune their algorithm accordingly.  There is also an inverse relationship between optimal variance and AOAR (see Table \ref{tab:results}) implying that when dealing with sub-optimal acceptance functions, the algorithm seeks larger jumps. The computational cost of the Bernoulli factory we present for $\alpha_r^{\text{R}}$  in Appendix~\ref{sec:BF} increases with $r$. Given the large jump in the optimal acceptance probability from $r = 1$ to $r = 2$, the development of more efficient Bernoulli factories is an important problem for future work. 
     % It depends on the proposal distribution and what information it utilizes to propose efficient moves. Same scaling factor for each function in the class $\mathcal{A}$ moreover implies that, given the algorithms are optimally tuned, each algorithm takes $O(d)$ steps to converge. Thus, in this aspect, all acceptance functions can be argued to be equally efficient. 
   
    % Doing this makes up for more rejections and enables the chain to explore the space effectively. 
    %it is desirable to introduce higher variance in proposals.
      % so as to meet the balance that will make the algorithm converge in $O(d)$ steps. 
      
   The assumption of starting from stationarity is a restrictive one. For MH with Gaussian proposals, the scaling factor of $1/d$ is still optimal when the algorithm is in the transient phase \citep{christensen2005,jourdain2014,kuntz2019}. The optimal acceptance probability may vary depending on the starting distribution. We envision similar results are viable for the general class of acceptance functions, and this is important future work.   Our results are limited to only Gaussian proposals and trivially decomposable target densities. Other proposal distributions may make use of the gradient of the target e.g. Metropolis-adjusted Langevin algorithm \citep{roberts1996} and Hamiltonian Monte Carlo \citep{duaneHMC1987}. In problems where $\alpha_{\text{MH}}$ cannot be used, the gradient of the target density is likely unavailable, thus limiting our attention to a Gaussian proposal is reasonable.
    % However, once the gradient is available, target distribution is also known and the MH accept-reject step can be implemented, hence, taking away the motivation for this paper. 
    On the other hand, generalizations to other target distributions is important. 
    % In this paper, we have considered a base case which was the starting point for MH acceptance function also in 1997. For this base case, we have illustrated how to extend the results for MH to any arbitrary acceptance function. 
    For MH algorithms, \cite{bedard2008,sher:rob:2009} relax the independence assumption, while \cite{roberts2001} relax the identically distributed assumption. Additionally, \cite{yang2020} present a proof of weak convergence for MH for more general targets, and \cite{schmon2021optimal} provide optimal scaling results for general Bayesian targets using large-sample asymptotics. In these situations, extensions to other acceptance probabilities are similarly possible. Additionally, we encourage future work  in optimal scaling to leverage our proof technique to demonstrate results for the wider class of acceptance probabilities.

\section{Acknowledgements}
The authors thank the referees and the editor for their comments that helped improve the presentation of the paper. Dootika Vats is supported by SERB grant: SPG/2021/001322. Krzysztof {\L}atuszy\'nski is supported by the Royal Society through the  Royal Society University Research Fellowship.  Gareth Roberts is supported by the EPSRC grants:  CoSInES (EP/R034710/1) and Bayes for Health (EP/R018561/1).

\appendix
% \appendixpage

\section{Proof of Theorem \ref{thm:main}}  % (fold)
\label{sec:proof}
    The proof is structurally similar to the seminal work of \cite{roberts1997weak}, in that we will show that the generator of the sped-up process, $\mathbf{Z}^d$, converges to the generator of an appropriate Langevin diffusion. Define the discrete-time generator of $\bs{Z}^d$ as,
    \begin{equation}
    \label{eq:gen}
        G_d V(\bs{x}^d) = d \cdot \mathbb{E}_{\bs{Y}^d}\left[(V(\bs{Y}^d) - V(\bs{x}^d))\alpha(\bs{x}^d, \bs{Y}^d)\right],
    \end{equation}
    for all those $V$ for which the limit exists.  
    % Since the limiting process is continuous-time, it is convenient that we interpret the sped up algorithm \eqref{eq:zt} as a Poisson jump process with jump rate $d$. The generator for this jump process remains identical to $G_d$. 
    Since, interest is in the first component of $\bs{Z}^d$, we consider only those $V$ which are functions of the first component only. Now, define the generator of the limiting Langevin diffusion process with speed measure $h_{\alpha}(l)$ as,
    \begin{equation}
    \label{eq:lang_gen}
        GV(x) = h_{\alpha}(l)\left[ \frac{1}{2}V''(x) + \frac{1}{2}\frac{d}{dx}(\log f)(x)V'(x) \right].
    \end{equation}
    The unique challenge in our result is identifying the speed measure $h_{\alpha}(l)$ for a general acceptance function $\alpha \in \mathcal{A}$. Proposition~\ref{prop:odd} is a key result that helps us obtain a form of $h_{\alpha}(l)$ without resorting to approximations.

    To prove Theorem \ref{thm:main}, we will show that there are events $F_d \subseteq \mathbb{R}^d$ such that for all $t$,
    \[
        \mathbb{P}[\bs{Z}_s^d \in F_d,\  0 \le s \le t] \to 1 \text{ as } d \to \infty \quad \text{    and}
    \]
    \[
        \lim_{d \to \infty}\sup_{\bs{x}^d \in F_d}|G_dV(\bs{x}^d) - GV(x_1^d)| = 0\,,
    \]
    for a suitably large class of real-valued functions $V$. Moreover, due to conditions of Lipschitz continuity on $f'/f$, a core for the generator $G$ has domain $C_c^{\infty}$, the class of infinitely differentiable functions with compact support \citep[Theorem~2.1, Chapter 8]{ethier2009}. Thus, we can limit our attention to only those $V \in C_c^{\infty}$ that are a function of the first component.

    Consider now the setup of Theorem~\ref{thm:main}. Let $w = \log f$ and $\alpha \in \mathcal{A}$ with the balancing function $g_{\alpha}$. Let $w'$ and $w''$ be the first and second derivatives of $w$ respectively. Define the sequence of sets $\{F_d \subseteq \mathbb{R}^d, d > 1 \}$ by,
    % \[
       % 
    % \]
    \begin{align*}
     F_d = \big\{|R_d(x_2, \dots, x_d) - I| &< d^{-1/8}  \big\} \cap \big\{|S_d(x_2, \dots, x_d) - I| < d^{-1/8} \big\} \quad \text{where}, \\
     R_d(x_2, \dots, x_d) &= \frac{1}{d-1}\sum_{i=2}^d \ [\log(f(x_i))']^2 = \frac{1}{d-1}\sum_{i=2}^d \ [w'(x_i)]^2 \quad \text{     and} \\
      S_d(x_2, \dots, x_d) &= \frac{-1}{d-1} \sum_{i=2}^d \ [\log(f(x_i))''] = \frac{-1}{d-1}\sum_{i=2}^d \ [w''(x_i)]  \,.
    \end{align*}
    The following results from \cite{roberts1997weak} will be needed.
    \begin{lemma}[\cite{roberts1997weak}]
    \label{lemm:fd}
       Let Assumption~\ref{assum:f} hold. If $\bs{X}^d_0 \sim \bs{\pi}_d$ for all $d$, then, for a fixed $t$, 
        $
            \mathbb{P}[\bs{Z}_s^d \in F_d,\  0 \le s \le t] \to 1 \text{ as } d \to \infty\,.
        $
    \end{lemma}
    \begin{lemma}[\cite{roberts1997weak}]
        \label{lemm:wd}
        Let Assumption~\ref{assum:f} hold. Also, let
        \[
            W_d(x_1, \dots, x_d) = \sum_{i=2}^d\left( \frac{1}{2}w''(x_i)(Y_i - x_i)^2 +\frac{l^2}{2(d-1)}w'(x_i)^2 \right),
        \]
        where $Y_i \overset{\text{ind}}{\sim} N(x_i, \sigma^2_d)$, $i = 2, \dots, d$. Then, $ \sup_{\bs{x}^d \in F_d}\mathbb{E}\left[\left|W_d(\bs{x}^d)\right|\right] \to 0\,$.
        % \[
        %     \sup_{\bs{x}^d \in F_d}\mathbb{E}_{\bs{Y}^{d-}}\left[\left|W_d(\bs{x}^d)\right|\right] \to 0\,.
        % \]
    \end{lemma}
    % \begin{proof}
    %     See Lemma 2.3 in \cite{roberts1997weak}
    % \end{proof}
    \begin{lemma}[\cite{roberts1997weak}]
        \label{lemm:uni1}
    For $Y \sim N(x, \sigma^2_d)$ and $V \in C_c^{\infty}$,
    \begin{equation*}
        \limsup_{d \to \infty}\, \sup_{x \in \mathbb{}R} d|\mathbb{E}[V(Y) - V(x)]| < \infty\,.
    \end{equation*}
    \end{lemma}
    % \begin{proof}
    %     See Lemma 2.5 in \cite{roberts1997weak}
    % \end{proof}
    For the following proposition, we will utilize the property \eqref{eq:detailed_bal} imposed on $\mathcal{A}$. This proposition is the key to obtaining our main result in such generality.
    \begin{propo}
        \label{prop:odd}
        Let $X \sim N(-\theta/2, \theta)$ for some $\theta > 0$. Let $\alpha \in \mathcal{A}$ with the corresponding balancing function $g_\alpha$. Then $\mathbb{E}\left[ Xg_{\alpha}(e^X)\right] = 0$.
    \end{propo}
    \begin{proof}
        We have,
        \[
            \big|\mathbb{E}\left[Xg_{\alpha}(e^X) \right]\big| \le \mathbb{E}\left[|Xg_{\alpha}(e^X)| \right] \le \mathbb{E}\left[|X| \right] < \infty;
        \]
        the second inequality follows from the assumption that $g_{\alpha}$ lies in [0,1]. Hence, the expectation exists and is equal to the integral,
        \[
            \int_{\mathbb{R}} x\,g_{\alpha}\left( e^{x}\right) \frac{1}{\sqrt{2\pi\theta}} \exp\left \{ \frac{-(x+\theta/2)^2}{2\theta}\right\} dx =: \int_{\mathbb{R}} h(x) dx \,.
        \]
        Observe that, using \eqref{eq:detailed_bal},
        \begin{align*}
            h(-x) &= -x\,g_{\alpha}\left( e^{-x}\right) \frac{1}{\sqrt{2\pi \theta}} \exp\left \{ \frac{-(-x + \theta/2)^2}{2\theta}\right\} \\
            &=  -x\,g_{\alpha}\left( e^{-x}\right) \frac{1}{\sqrt{2\pi\theta}} \exp\left \{ \frac{-1}{2\theta}\left( x^2 + \frac{\theta^2}{4} - x\theta \right)\right\} \\
            &= -x e^{-x} g_{\alpha}(e^x) \frac{1}{\sqrt{2\pi\theta}} \exp\left \{ \frac{-1}{2\theta}\left( x^2 + \frac{\theta^2}{4} - x\theta \right)\right\} \\
            % &= -x\,g_{\alpha}\left( e^{x}\right)\frac{1}{\sqrt{2\pi\theta}} \exp\left \{-x + \frac{-1}{2\theta}\left( x^2 + \frac{\theta^2}{4} - x\theta \right)\right\} \\
            &= -x\,g_{\alpha}\left( e^{x}\right)\frac{1}{\sqrt{2\pi\theta}} \exp\left \{\frac{-1}{2\theta}\left( x^2 + \frac{\theta^2}{4} + x\theta \right)\right\} \\
            &= -x\,g_{\alpha}\left( e^{x}\right) \frac{1}{\sqrt{2\pi\theta}} \exp\left \{ \frac{-(x+\theta/2)^2}{2\theta}\right\} \\
            &= - h(x).
        \end{align*}
       Hence, the result follows.
    \end{proof}

    \begin{lemma}
        \label{lemm:uni2}
        Suppose $V \in C_c^{\infty}$ is restricted to only the first component of $\bs{Z}^d$. Then,
        \[
            \sup_{\bs{x^d} \in F_d} |G_dV(\bs{x}^d) - GV(x_1^d)| \to 0 \text{ as } d \to \infty.
        \]
    \end{lemma}
    \begin{proof}
        In the expression for $G_dV(\bs{x}^d)$ given in \eqref{eq:gen}, we can decompose the proposal $\bs{Y}^d$ into $(Y_1^d, \bs{Y}^{d-})$ and thus rewrite the expectation as follows,
        \begin{equation}
            \label{eq:gen1}
            G_dV(\bs{x}^d) = d\mathbb{E}_{Y_1^d}\left[\left(V(Y^d_1) - V(x^d_1)\right)\mathbb{E}_{\bs{Y}^{d-}}\hspace{-.2cm}\left[ \alpha(\bs{x}^d, \bs{Y}^d) \mid Y_1^d \right]\right].
        \end{equation}
        Let $E^{d, \alpha}$ denote the inner expectation in \eqref{eq:gen1} and define $E_{lim}^{d, \alpha}$ as,
        \begin{equation}
        \label{eq:elim}
            \begin{aligned}
            E_{lim}^{d, \alpha} = \mathbb{E}_{\bs{Y}^{d-}} \hspace{-.2cm}\left[ g_{\alpha}\left( \exp\left\{ \log\dfrac{f(Y_1^d)}{f(x_1^d)} + \displaystyle \sum_{i=2}^d\left( w'(x_i^d)(Y_i^d - x_i^d) - \frac{l^2w'(x_i^d)^2}{2(d-1)}\right) \right\} \right) \bigg| Y_1^d \right].
            \end{aligned}
        \end{equation}
        Also, a Taylor series expansion of $w$ about $x_i^d$ for $i = 2, \dots, d$ gives,
        \begin{align*}
        %\label{eq:eda}
            E^{d, \alpha} &= \mathbb{E}_{\bs{Y}^{d-}} \hspace{-.2cm}\left[ g_{\alpha}\left( \exp\left\{ \log\frac{f(Y_1^d)}{f(x_1^d)} + \sum_{i=2}^d w'(x_i^d)(Y_i^d - x_i^d) \right. \right. \right. \\
            &\quad \quad \quad \quad \left. \left. \left. + \frac{1}{2}w''(x_i^d)(Y_i^d - x_i^d)^2 + \frac{1}{6}w'''(Z_i)(Y_i^d - x_i^d)^3 \right\} \right) \bigg| Y_1^d \right] 
        \end{align*}
        for $Z_i$ lying between $x_i^d$ and $Y_i^d$. Hence, the triangle inequality and Lipschitz continuity of $g(e^z)$ gives, for some Lipschitz constant $K < \infty$, 
        % Following \cite{roberts1997weak},
        % \begin{equation}
        % \label{eq:phi}
        %     \varphi(d) := \sup_{\bs{x}^d \in F_d}|E^{d, \alpha} - E_{lim}^{d, \alpha}| \to 0 \text{ as } d \to \infty.
        % \end{equation}
        % \cite{roberts1997weak} proved it for $g_{\alpha} = g_{\text{MH}}$. Same steps can in fact be adapted for an arbitrary function $g_{\alpha}$.
        % Then given $\bs{x}^d$, $E^{d,\alpha}_{lim}$ is a function of $Y_1^d$ alone, {\it to wit},
        % \begin{equation*}
        %     M_{d, \alpha}(Y_1^d) := E^{d,\alpha}_{lim} = \mathbb{E}\left[ g_{\alpha}(e^{B_d})\right],
        % \end{equation*} 
        % where $B_d \sim N(\mu_d, \Sigma_d)$ with $ \mu_d = \log (f(Y_1^d) / f(x_1^d)) - l^2R_d/2$ and $\Sigma_d = l^2R_d$.
        \begin{align}
        \label{eq:phid}
            |E^{d, \alpha} - E^{d,\alpha}_{lim}| &\le K\mathbb{E}_{\bs{Y}^{d-}}\hspace{-.2cm}\left[ \left| \sum_{i=2}^d \frac{1}{2}w''(x_i^d)(Y_i^d - x_i^d)^2 + \frac{1}{6}w'''(Z_i)(Y_i^d - x_i^d)^3 + \frac{l^2w'(x_i^d)^2}{2(d-1)} \right| \right] \notag \\
            &\le K\mathbb{E}_{\bs{Y}^{d-}}\hspace{-.2cm}\left[\left|W_d(\bs{x}^d)\right|\right] + K\sup_{z\in \mathbb{R}}|w'''(z)|\frac{l^3}{(d-1)^{1/2}},
        \end{align}    
        where $W_d(\bs{x}^d)$ is as defined in Lemma~\ref{lemm:wd}. From Lemma~\ref{lemm:wd}, Lemma~\ref{lemm:uni1} and \eqref{eq:phid},
        %Using $E_{lim}^{d, \alpha}$ from \eqref{eq:elim}, define,
        % \begin{equation*}
        %     G^*_dV(\bs{x^d}) = d\mathbb{E}_{Y_1^d}\left[\left(V(Y^d_1) - V(x^d_1)\right)E_{lim}^{d, \alpha}\right].
        % \end{equation*}
        %Then \eqref{eq:phid} together with Lemma~\ref{lemm:wd} and Lemma~\ref{lemm:uni1} leads to
        \begin{equation}
        \label{eq:gstar}
            \sup_{\bs{x}^d \in F_d} \left|G_dV(\bs{x}^d) - d\mathbb{E}_{Y_1^d}\left[\left(V(Y^d_1) - V(x^d_1)\right)E_{lim}^{d, \alpha}\right]\right| \to 0 \text{ as } d \to \infty.
        \end{equation}
\begin{comment}
        Since $V$ has a compact support, $S$(say), there exists $K < \infty$ such that $|g^{(i)}(x)|$, $|V^{(i)}(x)| \le K$ for $x \in S, \ \ i = 1, 2, 3$ \citep{roberts1997weak}. Then from Lemma \ref{lemm:elim} and \ref{lemm:uni1}, we have that
        \begin{align*}
            \sup_{\bs{x}^d \in F_d} |G_dV(\bs{x}^d) - G^*_dV(\bs{x}^d)| &=  \sup_{\bs{x}^d \in F_d} \left|d\mathbb{E}_{Y_1}\left[\left(V(Y^d_1) - V(x^d_1)\right)(E^{d, \alpha} - E_{lim}^{d, \alpha})\right]\right|, \\
            &\le d\sup_{\bs{x}^d \in F_d} \mathbb{E}_{Y_1}\left[\left|V(Y^d_1) - V(x^d_1)\right|\left|E^{d, \alpha} - E_{lim}^{d, \alpha}\right|\right], \\
            &\le \varphi(d) d\mathbb{E}_{Y_1}\left[|V(Y^d_1) - V(x^d_1)|\right], \\
            & \to 0 \text{ as } d \to \infty.
        \end{align*}
\end{comment}
        Now let $\epsilon(y) = \log f(y) - \log f(x_1^d)$. Also from \eqref{eq:elim}, it is clear that given $\bs{x}^d$, $E^{d,\alpha}_{lim}$ is a function of $Y_1^d$ alone, {\it to wit},
        \begin{equation}
        \label{eq:me}
            (M_{d, \alpha}\circ\epsilon)(Y_1^d) := E^{d,\alpha}_{lim} = \mathbb{E}\left[ g_{\alpha}(e^{B_d})\right],
        \end{equation} 
        where $B_d \sim N(\mu_d, \Sigma_d)$ with $ \mu_d = \epsilon(Y_1^d) - l^2R_d/2$ and $\Sigma_d = l^2R_d$.
        % Hence, it is enough to consider the sequence $G^*_d$ instead of $G_d$. and define $M_{d, \alpha}(\cdot)$ such that,
        % \begin{equation}
        %     \label{eq:me}
        %     M_{d, \alpha}(\epsilon(Y_1^d)) = E_{lim}^{d, \alpha} = \mathbb{E}\left[ g_{\alpha}(e^{B_d})\right],
        % \end{equation}
        % where $B_d \sim N(\mu_d, \Sigma_d)$ with $ \mu_d = \epsilon(Y_1^d) - l^2R_d/2$, $\Sigma_d = l^2R_d$, and $E_{lim}^{d, \alpha}$ depends on $\epsilon$ only through $\mu_d$. 
        Thus by \eqref{eq:phid}, it is enough to consider the asymptotic behaviour of,
        \[
            d\mathbb{E}_{Y_1^d}\left[\left(V(Y^d_1) - V(x^d_1)\right)M_{d, \alpha}(\epsilon(Y_1^d))\right].
        \]
        Let $N_{d, \alpha} = M_{d, \alpha} \circ \epsilon$ and apply Taylor series expansion on the inner term to obtain,
        \begin{align*}
            &\left(V(Y^d_1) - V(x^d_1)\right)M_{d, \alpha}(\epsilon(Y_1^d)) \\
            % & \quad =  \left(V(Y^d_1) - V(x^d_1)\right)N_{d, \alpha}(Y_1) \\
            &\quad =  \left( V'(x_1^d)(Y_1^d - x_1^d) + \frac{1}{2}V''(x_1^d)(Y_1^d - x_1^d)^2 + \frac{1}{6}V'''(K_d)(Y_1^d - x_1^d)^3\right) \\
            &\quad \quad  \times \left( N_{d, \alpha}(x_1^d) + N'_{d, \alpha}(x_1^d)(Y_1^d - x_1^d) + \frac{1}{2}N''_{d, \alpha}(L_d)(Y_1^d - x_1^d)^2\right)
        \end{align*}
        where $K_d, L_d \in [Y_1^d, x_1^d]$ or $[x_1^d, Y_1^d]$ and,
        \begin{align*}
            N_{d, \alpha}(x_1^d) &= M_{d, \alpha}(\epsilon(x_1^d)) = M_{d, \alpha}\left(\log \frac{f(x_1^d)}{f(x_1^d)}\right) = M_{d, \alpha}(0) \numberthis \label{eq:nd}\\
            N'_{d, \alpha}(x_1^d) &= M'_{d, \alpha}(\epsilon(x_1^d))\epsilon'(x_1^d) = M'_{d, \alpha}(0)w'(x_1^d) \,.
        \end{align*}
        Now, for all $d$,
        \begin{align*}
            M_{d, \alpha}(\epsilon) &= \mathbb{E}\left[ g_{\alpha}(e^{B_d})\right] 
            % = \int_{\mathbb{R}} g_{\alpha}(e^b) \frac{1}{\sqrt{2\pi\Sigma_d}} \exp\left \{ \frac{-(b - \mu_d)^2}{2\Sigma_d}\right\} db \\
            = \int_{\mathbb{R}} g_{\alpha}(e^b) \frac{1}{\sqrt{2\pi l^2R_d}} \exp\left \{ \frac{-(b - \epsilon + l^2R_d/2)^2}{2l^2R_d}\right\} db. \\
            \text{So, } M_{d, \alpha}(0) &= \int_{\mathbb{R}} g_{\alpha}(e^b) \frac{1}{\sqrt{2\pi l^2R_d}} \exp\left \{ \frac{-(b + l^2R_d/2)^2}{2l^2R_d}\right\} db. \\
            \text{Also, }M_{d, \alpha}'(\epsilon) &= \frac{d}{d\epsilon} \left( \int_{\mathbb{R}} g_{\alpha}(e^b) \frac{1}{\sqrt{2\pi l^2R_d}} \exp\left \{ \frac{-(b - \epsilon + l^2R_d/2)^2}{2l^2R_d}\right\} db \right).
        \end{align*}
        Derivatives and integral are exchanged due to the dominated convergence theorem. So,
        \begin{align*}
            M_{d, \alpha}'(\epsilon) 
            % &= \int_{\mathbb{R}} \frac{d}{d\epsilon} \left(  g_{\alpha}(e^b) \frac{1}{\sqrt{2\pi l^2R_d}} \exp\left \{ \frac{-(b - \epsilon + l^2R_d/2)^2}{2l^2R_d}\right\}  \right) db, \\
            &= \int_{\mathbb{R}} g_{\alpha}(e^b) \frac{1}{\sqrt{2\pi l^2R_d}} \left( \frac{2(b - \epsilon + l^2R_d/2)}{2l^2R_d} \right) \exp\left \{ \frac{-(b - \epsilon + l^2R_d/2)^2}{2l^2R_d}\right\} db. \\
            \text{So, } M'_{d, \alpha}(0) &= \int_{\mathbb{R}} g_{\alpha}(e^b) \frac{1}{\sqrt{2\pi l^2R_d}} \left( \frac{(b + l^2R_d/2)}{l^2R_d} \right) \exp\left \{ \frac{-(b + l^2R_d/2)^2}{2l^2R_d}\right\} db \\
            &= \frac{1}{l^2R_d}\int_{\mathbb{R}} b\,g_{\alpha}(e^b) \frac{1}{\sqrt{2\pi l^2R_d}} \exp\left \{ \frac{-(b + l^2R_d/2)^2}{2l^2R_d}\right\} db \\
            &\hspace{50pt}+ \frac{1}{2}\int_{\mathbb{R}} g_{\alpha}(e^b) \frac{1}{\sqrt{2\pi l^2R_d}}\exp\left \{ \frac{-(b + l^2R_d/2)^2}{2l^2R_d}\right\} db \\
            &= \frac{1}{2} M_{d, \alpha}(0)\,,
        \end{align*}
         where the first term vanishes due to Proposition \ref{prop:odd}. Hence, for all $d$,
        \begin{equation}
            \label{eq:md0}
            2M_{d, \alpha}'(0) = M_{d, \alpha}(0) = \int_{\mathbb{R}} g_{\alpha}(e^b) \frac{1}{\sqrt{2\pi l^2R_d}} \exp\left \{ \frac{-(b + l^2R_d/2)^2}{2l^2R_d}\right\} db.
        \end{equation}
        Now, we plug the expressions obtained above into the Taylor series expansion of $\left(V(Y^d_1) - V(x^d_1)\right)M_{d, \alpha}(\epsilon(Y_1^d))$. The rest of the proof, with the help of Assmuption \ref{assum:f}, follows similarly as in Lemma 2.6, \cite{roberts1997weak}.
    \end{proof}
    \begin{proof}[\bf{Proof of Theorem \ref{thm:main}}]
        From Lemma \ref{lemm:uni2}, we have uniform convergence of generators on the sequence of sets with limiting probability 1. And so by Corollary 8.7, Chapter 4 of \cite{ethier2009}, we have the required result of weak convergence  \citep[the condition that $C_c^{\infty}$ separates points was verified by][]{roberts1997weak}.
    \end{proof}

\section{Proof of Corollary \ref{corr:1}}
    % We are in the same setup as in the previous section and all the nice properites required from $w$ and $\alpha$ hold. 
    \begin{lemma}
    \label{lemm:elim2}
        Let $E^{d, \alpha}$ be the inner expectation in \eqref{eq:gen1} and $E_{lim}^{d, \alpha}$ be from \eqref{eq:elim}. Then,
        \[ 
            \mathbb{E}_{\bs{\pi}_d} \left[ \mathbb{E}_{Y_1} \left[ E^{d, \alpha} - E_{lim}^{d, \alpha}\bigg| \, \bs{x}^d \right]\right] \to 0 \ \ \ \ \ \text{as } d \to \infty.
        \]
    \end{lemma}
    \begin{proof}
        Consider,
        \begin{align*}
            \left| \mathbb{E}_{\bs{\pi}_d} \left[ \mathbb{E}_{Y_1^d} \left[ E^{d, \alpha} - E_{lim}^{d, \alpha}\bigg| \, \bs{x}^d \right]\right] \right|  
            &\le \left| \mathbb{E}_{\bs{\pi}_d} \left[ \mathbb{E}_{Y_1^d} \left[ E^{d, \alpha} - E_{lim}^{d, \alpha}\bigg| \, \bs{x}^d \in F_d \right]\right] P(\bs{x}^d \in F_d) \right| \\
            & \quad + \left| \mathbb{E}_{\bs{\pi}_d} \left[ \mathbb{E}_{Y_1^d} \left[ E^{d, \alpha} - E_{lim}^{d, \alpha}\bigg| \, \bs{x}^d \in F^C_d \right]\right] P(\bs{x}^d \in F^C_d) \right|.
        \end{align*}
        Second term goes to $0$ since the expectation is bounded and by construction $P(\bs{x}^d \in F_d^C)\to 0$ as $d \to \infty$. Also, following \cite{roberts1997weak},
        \begin{equation*}
            \sup_{\bs{x}^d \in F_d}|E^{d, \alpha} - E_{lim}^{d, \alpha}| \to 0 \text{ as } d \to \infty.
        \end{equation*}
        Then,
        \begin{align*}
            &\left| \mathbb{E}_{\bs{\pi}_d} \hspace{-.1cm}\left[ \mathbb{E}_{Y_1^d} \hspace{-.1cm}\left[ E^{d, \alpha} - E_{lim}^{d, \alpha}\bigg| \, \bs{x}^d \in F_d \right]\right]  P(\bs{x}^d \in F_d) \right| \\
            % &\le \left| \mathbb{E}_{\bs{\pi}_d} \hspace{-.1cm}\left[ \mathbb{E}_{Y_1} \hspace{-.1cm}\left[ E^{d, \alpha} - E_{lim}^{d, \alpha}\bigg| \, \bs{x}^d \in F_d \right]\right]\right|,  \\
            % & \quad \le \mathbb{E}_{\bs{\pi}_d} \left[ \mathbb{E}_{Y_1} \left[\left| E^{d, \alpha} - E_{lim}^{d, \alpha}\right|\bigg| \, \bs{x}^d \in F_d \right]\right], \\
            & \quad\le \mathbb{E}_{\bs{\pi}_d} \left[ \mathbb{E}_{Y_1^d} \left[\sup_{\bs{x^d} \in F_d}\left| E^{d, \alpha} - E_{lim}^{d, \alpha}\right|\bigg| \, \bs{x}^d \in F_d \right]\right] \to 0\,.
        \end{align*}
    \end{proof}
    \begin{proof}[{\bf Proof of Corollary \ref{corr:1}}]
        Consider equation \eqref{eq:me}. Using Taylor series approximation of second order around $x_1$,
        %  we get 
        % \[
        %     N_{d, \alpha}(Y_1) = N_{d, \alpha}(x_1) + N'_{d, \alpha}(x_1)(Y_1 - x_1) + \frac{1}{2}N''_{d, \alpha}(W_1)(Y_1 - x_1)^2,
        % \]
        \begin{equation*}
            \mathbb{E}_{Y_1^d}[E_{lim}^{d, \alpha}] = \mathbb{E}[N_{d, \alpha}(Y_1^d)] = N_{d, \alpha}(x_1^d) +  \frac{1}{2}N''_{d, \alpha}(W_{d,1})\frac{l^2}{d-1}\,.
        \end{equation*}
        where $W_{d,1} \in [x_1^d, Y_1^d]$ or $[Y_1^d, x_1^d]$. 
        % 
        % Using expression obtained for $N_{d, \alpha}(x_1)$ in (\ref{eq:nd}), we get
        % \[
        %     \mathbb{E}_{Y_1}[E_{lim}^{d, \alpha}] = M_{d, \alpha}(0) + \frac{1}{2}N''_{d, \alpha}(W_1)\frac{l^2}{d-1}.
        % \]
        Since $N''$ is bounded \citep{roberts1997weak}, 
        \begin{align*}
            \alpha(l) 
            % &= \lim_{d \to \infty} \mathbb{E}_{\bs{\pi}_d} \, [\alpha(\bs{X}^d, \bs{Y}^d)],\\
            % &=  \lim_{d \to \infty}  \mathbb{E}_{\bs{\pi}_d} \left[ \mathbb{E}_{\bs{Y}^d} \left[ \alpha(\bs{X}^d, \bs{Y}^d) \bigg| \, \bs{x}^d\right]\right], \\
            &= \lim_{d \to \infty}  \mathbb{E}_{\bs{\pi}_d} \left[ \mathbb{E}_{Y_1^d} \left[ \mathbb{E}_{\bs{Y}^{d-}}\left[\alpha(\bs{X}^d, \bs{Y}^d) \bigg| Y_1^d, \bs{x}^d\right] \bigg| \, \bs{x}^d \right]\right] \\
            &= \lim_{d \to \infty}  \mathbb{E}_{\bs{\pi}_d} \left[ \mathbb{E}_{Y_1^d} \left[E_{lim}^{d, \alpha} + E^{d, \alpha} - E_{lim}^{d, \alpha}\bigg| \bs{x}^d \right]\right].
        \end{align*}
        As all expectations exist, we can split the inner expectation and use Lemma \ref{lemm:elim2}, so that
        \begin{align*}
            \alpha(l) &= \lim_{d \to \infty}  \mathbb{E}_{\bs{\pi}_d} \left[ \mathbb{E}_{Y_1^d} \left[E_{lim}^{d, \alpha} \bigg| \bs{x}^d \right]\right] + \lim_{d \to \infty}  \mathbb{E}_{\bs{\pi}_d} \left[ \mathbb{E}_{Y_1^d} \left[ E^{d, \alpha} - E_{lim}^{d, \alpha}\bigg| \, \bs{x}^d \right]\right]\\
            &= \lim_{d \to \infty}  \mathbb{E}_{\bs{\pi}_d} \left[ M_{d, \alpha}(0) + \frac{1}{2}N''_{d, \alpha}(W_{d,1})\frac{l^2}{d-1}\right] \\
            &= \lim_{d \to \infty}  \mathbb{E}_{\bs{\pi}_d} \left[\int_{\mathbb{R}} g_{\alpha}(e^b) \frac{1}{\sqrt{2\pi l^2R_d}} \exp\left \{ \frac{-(b + l^2R_d/2)^2}{2l^2R_d}\right\} db \right] \\
            &= \int_{\mathbb{R}} g_{\alpha}(e^b) \frac{1}{\sqrt{2\pi l^2I}} \exp\left \{ \frac{-(b + l^2I/2)^2}{2l^2I}\right\} db = M_{\alpha}(l)\,.
        \end{align*}
        The last equality is by the law of large numbers and continuous mapping theorem. 
    \end{proof}

\section{Optimizing speed for Barker's acceptance} 
\label{sec:speed}
    We need to maximise $h_{\text{B}}(l) = l^2M_{\text{B}}(l)$.
    % where,
    % \[
    %     M_{\text{B}}(l) = \int_{\mathbb{R}} \frac{1}{1 + e^{-b}} \frac{1}{\sqrt{2\pi l^2I}} \exp\left \{ \frac{-(b + l^2I/2)^2}{2l^2I}\right\} db.
    % \]
   Let $I$ be fixed arbitrarily. 
   % Then multiplying and dividing $h_{\text{B}}(l)$ by I, we get,
    \[
        h_{\text{B}}(l) = \frac{1}{I}\cdot l^2I\cdot \int_{\mathbb{R}} \frac{1}{1 + e^{-b}} \frac{1}{\sqrt{2\pi l^2I}} \exp\left \{ \frac{-(b + l^2I/2)^2}{2l^2I}\right\} db.
    \]
    %Note that $h_{\text{B}}(l)$ is an even function. Hence, it suffices to talk about only the positive real line. Further, 
    For a fixed $I$, we can reparametrize the function by taking $\theta = l^2I$ and so maximizing $h_{\text{B}}(l)$ over positive $l$ will be equivalent to maximizing $h^1_{\text{B}}(\theta)$ over positive $\theta$ where,
    \begin{equation*}
       h^1_{\text{B}}(\theta) = \int_{\mathbb{R}} \frac{\theta}{1 + e^{-b}} \frac{1}{\sqrt{2\pi \theta}} \exp\left \{ \frac{-(b + \theta/2)^2}{2\theta}\right\} db.     
    \end{equation*}
    % \begin{figure}[]
    %     \centering
    %     \includegraphics[width = 3.5in]{plots/finding_l.pdf} \vspace{-.5cm}
    %     \caption{Optimizing $h(l)$.}
    %     \label{fig:speed}
    % \end{figure}
    We make the substitution $z = (b + \theta/2)/\sqrt{\theta}$ in the integrand 
%    \[
%        z = \frac{(b + \theta/2)}{\sqrt{\theta}},
%    \]
    to obtain 
    \[
        h_{\text{B}}^1(\theta) = \int_{\mathbb{R}} \frac{\theta}{1 + \exp\{-z\sqrt{\theta} + \theta/2\}} \frac{1}{\sqrt{2\pi}} e^{-z^2/2} dz = \mathbb{E}\left[ \frac{\theta}{1 + \exp\{-Z\sqrt{\theta} + \theta/2\}} \right],
    \]
    where the expectation is taken with respect to $Z \sim N(0,1)$. This expectation however is not available in closed form. However standard numerical integration routines yield the optimal value of $\theta$ to be $6.028$. This implies that the optimal value of $l$, say $l^*$, is approximately equal to,
    \[
        l^* \approx \frac{2.46}{\sqrt{I}} \ \ (\text{up to 2 decimal places}). 
    \]
    Using this $l^*$ yields an asymptotically optimal acceptance rate of approximately $0.158$.

\section{Bernoulli factory}
\label{sec:BF}

To sample events of probability $\alpha_B$, the \textit{two-coin} algorithm, an efficient Bernoulli factory, was presented in \cite{gonccalves2017b}. Generalizing this to a \textit{die-coin} algorithm, we present a Bernoulli factory for $\alpha_r^{\text{R}}$ for $r =2$; extensions to other $r$ can be done similarly. Let $\pi(x) = c_x p_x$ with $p_x \in [0,1]$ and $c_x > 0$. Then,
\[ 
\alpha_2^{\text{R}}(x,y)  =  \dfrac{\pi(y)^2 + \pi(x) \pi(y)}{\pi(y)^2 + \pi(x) \pi(y) + \pi(x)^2} = \dfrac{c_y^2p_y^2 + c_xp_x c_yp_y}{c_y^2p_y^2 + c_xp_x c_yp_y + c_x^2p_x^2}\,.
\]
\begin{algorithm}[htbp]
\caption{Die-coin algorithm for $\alpha_2^{\text{R}}(x,y)$}\label{alg:alpha2}
\begin{algorithmic}[1]
% \Procedure{}{}
\State Draw $D \sim $ Categorical$\left(\dfrac{c_y^2}{c_x^2 +c_xc_y +c_y^2}, \dfrac{c_xc_y}{c_y^2 +c_xc_y +c_x^2}, \dfrac{c_x^2}{c_y^2 +c_xc_y +c_x^2}  \right)$
\If {$D = 1$} 
\State Draw $C_1 \sim \text{Bern}(p_y^2)$
\If {$C_1 = 1$} {output 1} {\textbf{else} go back to Step 1}
\EndIf
% \If {$C_1 = 0$} {go back to Step 1}
% \EndIf
\EndIf
\If {$D = 2$} 
\State Draw $C_1 \sim \text{Bern}(p_xp_y)$
\If {$C_1 = 1$} {output 1} {\textbf{else} go back to Step 1}
\EndIf
\EndIf
\If {$D = 3$} 
\State Draw $C_1 \sim \text{Bern}(p_x^2)$
\If {$C_1 = 1$} {output 0} {\textbf{else} go back to Step 1}
\EndIf
\EndIf
% \EndProcedure
\end{algorithmic}
\end{algorithm}

\bibliographystyle{apalike}
\bibliography{references}

\begin{thebibliography}{}

\bibitem[Banterle et~al., 2019]{banterle2019accelerating}
Banterle, M., Grazian, C., Lee, A., and Robert, C.~P. (2019).
\newblock Accelerating {M}etropolis-{H}astings algorithms by delayed
  acceptance.
\newblock {\em Foundations of Data Science}, 1(2):103--128.

\bibitem[Barker, 1965]{barker1965}
Barker, A.~A. (1965).
\newblock {M}onte {C}arlo calculations of the radial distribution functions for
  a proton-electron plasma.
\newblock {\em Australian Journal of Physics}, 18(2):119--134.

\bibitem[B{\'e}dard, 2008]{bedard2008}
B{\'e}dard, M. (2008).
\newblock Optimal acceptance rates for {M}etropolis algorithms: {M}oving beyond
  0.234.
\newblock {\em Stochastic Processes and their Applications},
  118(12):2198--2222.

\bibitem[Billera and Diaconis, 2001]{billera2001}
Billera, L.~J. and Diaconis, P. (2001).
\newblock A geometric interpretation of the {M}etropolis-{H}astings algorithm.
\newblock {\em Statistical Science}, pages 335--339.

\bibitem[Brooks et~al., 2011]{brooks2011}
Brooks, S., Gelman, A., Jones, G., and Meng, X.-L. (2011).
\newblock {\em Handbook of {M}arkov chain {M}onte {C}arlo}.
\newblock CRC press.

\bibitem[Christensen et~al., 2005]{christensen2005}
Christensen, O.~F., Roberts, G.~O., and Rosenthal, J.~S. (2005).
\newblock Scaling limits for the transient phase of local
  {M}etropolis--{H}astings algorithms.
\newblock {\em Journal of the Royal Statistical Society: Series B (Statistical
  Methodology)}, 67(2):253--268.

\bibitem[Delmas and Jourdain, 2009]{del:jour:2009}
Delmas, J.-F. and Jourdain, B. (2009).
\newblock Does waste recycling really improve the multi-proposal
  {M}etropolis--{H}astings algorithm? {A}n analysis based on control variates.
\newblock {\em Journal of Applied Probability}, 46:938--959.

\bibitem[Doucet et~al., 2015]{doucet2015efficient}
Doucet, A., Pitt, M.~K., Deligiannidis, G., and Kohn, R. (2015).
\newblock Efficient implementation of {M}arkov chain {M}onte {C}arlo when using
  an unbiased likelihood estimator.
\newblock {\em Biometrika}, 102(2):295--313.

\bibitem[Duane et~al., 1987]{duaneHMC1987}
Duane, S., Kennedy, A.~D., Pendleton, B.~J., and Roweth, D. (1987).
\newblock Hybrid {M}onte {C}arlo.
\newblock {\em Physics Letters B}, 195(2):216--222.

\bibitem[Ethier and Kurtz, 1986]{ethier2009}
Ethier, S.~N. and Kurtz, T.~G. (1986).
\newblock {\em Markov processes: {C}haracterization and convergence}.
\newblock John Wiley \& Sons.

\bibitem[Gelman et~al., 1996]{gelman1996}
Gelman, A., Roberts, G.~O., and Gilks, W.~R. (1996).
\newblock Efficient {M}etropolis jumping rules.
\newblock {\em Bayesian Statistics}, 5:599--608.

\bibitem[Gon{\c{c}}alves et~al., 2017a]{gonccalves2017a}
Gon{\c{c}}alves, F.~B., {\L}atuszy{\'n}ski, K., and Roberts, G.~O. (2017a).
\newblock Barker's algorithm for {B}ayesian inference with intractable
  likelihoods.
\newblock {\em Brazilian Journal of Probability and Statistics},
  31(4):732--745.

\bibitem[Gon{\c{c}}alves et~al., 2017b]{gonccalves2017b}
Gon{\c{c}}alves, F.~B., {\L}atuszy{\'n}ski, K., and Roberts, G.~O. (2017b).
\newblock Exact {M}onte {C}arlo likelihood-based inference for jump-diffusion
  processes.
\newblock {\em arXiv preprint arXiv:1707.00332}.

\bibitem[Hastings, 1970]{hastings1970}
Hastings, W.~K. (1970).
\newblock {M}onte {C}arlo sampling methods using {M}arkov chains and their
  applications.
\newblock {\em Biometrika}, 57(1):97--109.

\bibitem[Herbei and Berliner, 2014]{herbei:berliner:2014}
Herbei, R. and Berliner, L.~M. (2014).
\newblock Estimating ocean circulation: an {MCMC} approach with approximated
  likelihoods via the {B}ernoulli factory.
\newblock {\em Journal of the American Statistical Association}, 109:944--954.

\bibitem[Jourdain et~al., 2014]{jourdain2014}
Jourdain, B., Leli{\`e}vre, T., and Miasojedow, B. (2014).
\newblock Optimal scaling for the transient phase of {M}etropolis {H}astings
  algorithms: {T}he longtime behavior.
\newblock {\em Bernoulli}, 20:1930--1978.

\bibitem[Kuntz et~al., 2019]{kuntz2019}
Kuntz, J., Ottobre, M., and Stuart, A.~M. (2019).
\newblock Diffusion limit for the random walk {M}etropolis algorithm out of
  stationarity.
\newblock In {\em Annales de l'Institut Henri Poincar{\'e}, Probabilit{\'e}s et
  Statistiques}, volume~55, pages 1599--1648. Institut Henri Poincar{\'e}.

\bibitem[{\L}atuszy{\'n}ski and Roberts, 2013]{latuszynski2013}
{\L}atuszy{\'n}ski, K. and Roberts, G.~O. (2013).
\newblock C{LT}s and asymptotic variance of time-sampled {M}arkov chains.
\newblock {\em Methodology and Computing in Applied Probability},
  15(1):237--247.

\bibitem[Menezes and Kabamba, 2014]{menezes2014}
Menezes, A.~A. and Kabamba, P.~T. (2014).
\newblock Optimal search efficiency of {B}arker's algorithm with an exponential
  fitness function.
\newblock {\em Optimization Letters}, 8(2):691--703.

\bibitem[Metropolis et~al., 1953]{metropolis1953}
Metropolis, N., Rosenbluth, A.~W., Rosenbluth, M.~N., Teller, A.~H., and
  Teller, E. (1953).
\newblock Equation of state calculations by fast computing machines.
\newblock {\em The {J}ournal of {C}hemical {P}hysics}, 21(6):1087--1092.

\bibitem[Meyn and Tweedie, 2012]{meyn2012}
Meyn, S.~P. and Tweedie, R.~L. (2012).
\newblock {\em Markov chains and {S}tochastic stability}.
\newblock Springer Science \& Business Media.

\bibitem[Mira, 2001]{mira2001}
Mira, A. (2001).
\newblock On {M}etropolis-{H}astings algorithms with delayed rejection.
\newblock {\em Metron}, 59(3-4):231--241.

\bibitem[Morina et~al., 2021]{morina2019bernoulli}
Morina, G., {\L}atuszy\'{n}ski, K., Nayar, P., and Wendland, A. (2021).
\newblock From the {B}ernoulli factory to a dice enterprise via perfect
  sampling of {M}arkov chains.
\newblock {\em Annals of Applied Probability, \rm{to appear.}}

\bibitem[Neal and Roberts, 2006]{neal:optimal:2006}
Neal, P. and Roberts, G.~O. (2006).
\newblock Optimal scaling for partially updating {MCMC} algorithms.
\newblock {\em The Annals of Applied Probability}, 16:475--515.

\bibitem[Peskun, 1973]{peskun1973}
Peskun, P.~H. (1973).
\newblock Optimum {M}onte-{C}arlo sampling using {M}arkov chains.
\newblock {\em Biometrika}, 60(3):607--612.

\bibitem[Robert and Casella, 2013]{robert2013}
Robert, C. and Casella, G. (2013).
\newblock {\em Monte {C}arlo {S}tatistical {M}ethods}.
\newblock Springer Science \& Business Media.

\bibitem[Roberts et~al., 1997]{roberts1997weak}
Roberts, G.~O., Gelman, A., and Gilks, W.~R. (1997).
\newblock {W}eak convergence and optimal scaling of random walk {M}etropolis
  algorithms.
\newblock {\em The Annals of Applied Probability}, 7(1):110--120.

\bibitem[Roberts and Rosenthal, 1998]{roberts1998}
Roberts, G.~O. and Rosenthal, J.~S. (1998).
\newblock Optimal scaling of discrete approximations to {L}angevin diffusions.
\newblock {\em Journal of the Royal Statistical Society: Series B (Statistical
  Methodology)}, 60(1):255--268.

\bibitem[Roberts and Rosenthal, 2001]{roberts2001}
Roberts, G.~O. and Rosenthal, J.~S. (2001).
\newblock Optimal scaling for various {M}etropolis-{H}astings algorithms.
\newblock {\em Statistical Science}, 16(4):351--367.

\bibitem[Roberts and Rosenthal, 2009]{roberts2009examples}
Roberts, G.~O. and Rosenthal, J.~S. (2009).
\newblock Examples of adaptive {MCMC}.
\newblock {\em Journal of Computational and Graphical Statistics}, 18:349--367.

\bibitem[Roberts and Tweedie, 1996]{roberts1996}
Roberts, G.~O. and Tweedie, R.~L. (1996).
\newblock Exponential convergence of {L}angevin distributions and their
  discrete approximations.
\newblock {\em Bernoulli}, 2:341--363.

\bibitem[Schmon et~al., 2021]{schmon2021large}
Schmon, S.~M., Deligiannidis, G., Doucet, A., and Pitt, M.~K. (2021).
\newblock Large-sample asymptotics of the pseudo-marginal method.
\newblock {\em Biometrika}, 108:37--51.

\bibitem[Schmon and Gagnon, 2021]{schmon2021optimal}
Schmon, S.~M. and Gagnon, P. (2021).
\newblock Optimal scaling of random walk {M}etropolis algorithms using
  {B}ayesian large-sample asymptotics.
\newblock {\em arXiv preprint arXiv:2104.06384}.

\bibitem[Sherlock and Roberts, 2009]{sher:rob:2009}
Sherlock, C. and Roberts, G.~O. (2009).
\newblock Optimal scaling of the random walk {M}etropolis on elliptically
  symmetric unimodal targets.
\newblock {\em Bernoulli}, 15:774--798.

\bibitem[Sherlock et~al., 2021]{sherlock2015}
Sherlock, C., Thiery, A.~H., and Golightly, A. (2021).
\newblock Efficiency of delayed-acceptance random walk {M}etropolis algorithms.
\newblock {\em The Annals of Statistics}, 49(5):2972--2990.

\bibitem[Sherlock et~al., 2015]{sherlock2015efficiency}
Sherlock, C., Thiery, A.~H., Roberts, G.~O., and Rosenthal, J.~S. (2015).
\newblock On the efficiency of pseudo-marginal random walk {M}etropolis
  algorithms.
\newblock {\em Annals of Statistics}, 43(1):238--275.

\bibitem[Smith, 2018]{smith:2018}
Smith, C.~J. (2018).
\newblock {\em Exact {M}arkov Chain {M}onte {C}arlo with Likelihood
  Approximations for Functional Linear Models}.
\newblock PhD thesis, The Ohio State University.

\bibitem[Vats et~al., 2019]{vats2019}
Vats, D., Flegal, J.~M., and Jones, G.~L. (2019).
\newblock Multivariate output analysis for {M}arkov chain {M}onte {C}arlo.
\newblock {\em Biometrika}, 106:321--337.

\bibitem[Vats et~al., 2021]{vats2020}
Vats, D., Gon{\c c}alves, F.~B., {\L}atuszy{\'n}ski, K., and Roberts, G.~O.
  (2021).
\newblock {Efficient {B}ernoulli factory {M}arkov chain {M}onte Carlo for
  intractable posteriors}.
\newblock {\em Biometrika}.
\newblock asab031.

\bibitem[Yang et~al., 2020]{yang2020}
Yang, J., Roberts, G.~O., and Rosenthal, J.~S. (2020).
\newblock Optimal scaling of random-walk {M}etropolis algorithms on general
  target distributions.
\newblock {\em Stochastic Processes and their Applications}, 130(10):6094 --
  6132.

\bibitem[Zanella et~al., 2017]{zanella2017dirichlet}
Zanella, G., B{\'e}dard, M., and Kendall, W.~S. (2017).
\newblock A {D}irichlet form approach to {MCMC} optimal scaling.
\newblock {\em Stochastic Processes and their Applications}, 127:4053--4082.

\end{thebibliography}

\end{document}